%% file: main.tex
\documentclass[11pt]{article}
\input{preamble}

\newcommand{\dtv}{d_{\mathrm TV}}
\newcommand{\dk}{d_{\mathrm K}}

\title{Faster and Sample Near-Optimal Algorithms for Proper Learning Mixtures of Gaussians}

\author {
Constantinos Daskalakis\thanks{Supported by a Sloan Foundation Fellowship, a Microsoft Research Faculty Fellowship, and NSF Award CCF-0953960 (CAREER) and CCF-1101491.}\\
EECS, MIT \\
\tt{costis@mit.edu}
\and
Gautam Kamath\thanks{Part of this work was done while the author was supported by an Akamai Presidential Fellowship.}\\
EECS, MIT\\
\tt{g@csail.mit.edu}
}
\begin{document}
\addtocounter{page}{-1}
\maketitle
\thispagestyle{empty}
\begin{abstract}
We provide an algorithm for properly learning mixtures of two single-dimensional Gaussians without any separability assumptions. Given $\tilde{O}(1/\varepsilon^2)$ samples from an unknown mixture, our algorithm outputs a mixture that is $\varepsilon$-close in total variation distance, in time $\tilde{O}(1/\varepsilon^5)$. Our sample complexity is optimal up to logarithmic factors, and significantly improves upon both Kalai et al.~\cite{KalaiMV10}, whose algorithm has a prohibitive dependence on~$1/\varepsilon$, and Feldman et al.~\cite{FeldmanOS06}, whose algorithm requires bounds on the mixture parameters and depends pseudo-polynomially in these parameters.

One of our main contributions is an improved and generalized algorithm for selecting a good candidate distribution from among competing hypotheses. Namely, given a collection of $N$ hypotheses containing at least one candidate that is $\varepsilon$-close to an unknown distribution, our algorithm outputs a candidate which is $O(\varepsilon)$-close to the  distribution.
The algorithm requires ${O}(\log{N}/\varepsilon^2)$ samples from the unknown distribution and ${O}(N \log N/\varepsilon^2)$ time, which improves previous such results (such as the Scheff\'e estimator) from a quadratic dependence of the running time on $N$ to quasilinear. Given the wide use of such results for the purpose of hypothesis selection, our improved algorithm implies immediate improvements to any such use.
\end{abstract}

\newpage

\input{introduction}

\input{preliminaries}
\input{tvbounds}
\input{kolmogorov}

\input{robust}
\input{outline}
\input{mixing_weight}
\input{means}
\input{devs}

\input{lastcomp}

\input{puttingtogether}
\input{tournament}

\bibliographystyle{alpha}
\bibliography{biblio}
\appendix
\input{gridding}

\input{proofs}
\input{intervalpartition}
\input{closestpointkol}
\input{tournament_appendix}

\end{document}

%% file: preamble.tex
\usepackage{ifthen}
\usepackage{mdwlist}
\usepackage{fullpage}
\usepackage{amsmath,amssymb,amsfonts,amsthm}
\usepackage{bm}
\usepackage{enumitem}
\usepackage{graphicx}
\usepackage{xspace}
\usepackage{verbatim}
\usepackage{algorithm}
\usepackage{algpseudocode}

\def\d{\delta}

\def\ve{\varepsilon}

\def\g{\gamma}

\def\m{\mu}

\def\p{\pi}

\def\s{\sigma}

\newcommand{\erf}{\mbox{\text erf}}
\newcommand{\erfc}{\mbox{\text erfc}}

\newcommand{\prob}[2][]{\text{\bf Pr}\ifthenelse{\not\equal{}{#1}}{_{#1}}{}\!\left[#2\right]}
\newcommand{\expect}[2][]{\text{\bf E}\ifthenelse{\not\equal{}{#1}}{_{#1}}{}\!\left[#2\right]}

\newtheorem{theorem}{Theorem}

\newtheorem{remark}{Remark}
\newtheorem{fact}[theorem]{Fact}
\newtheorem{lemma}[theorem]{Lemma}
\newtheorem{claim}{Claim}

\newtheorem{definition}{Definition}
\newtheorem{corollary}{Corollary}

\newcommand{\ignore}[1]{}
\newtheorem{proposition}[theorem]{Proposition}

\newenvironment{prevproof}[2]{\noindent {\em {Proof of {#1}~\ref{#2}:}}}{$\hfill\qed$\vskip \belowdisplayskip}

\newcommand{\bg}[1]{\medskip\noindent{\bf #1}}



%% file: introduction.tex
\section{Introduction} \label{sec:intro}

Learning mixtures of Gaussian distributions is one of the most fundamental problems in Statistics, with a multitude of applications in the natural and social sciences, which has recently received considerable attention in Computer Science literature. Given independent samples from an unknown mixture of Gaussians, the task is to `learn' the underlying mixture. 

\smallskip In one version of the problem, `learning' means estimating the {\em parameters} of the mixture, that is the mixing probabilities as well as the parameters of each constituent Gaussian. The most popular heuristic for doing so is running the EM algorithm on samples from the mixture~\cite{Dempster-Laird-Rubin77}, albeit no rigorous guarantees are known for it in general. 

A line of research initiated by Dasgupta~\cite{DasguptaFOCS99, AroraKannanSTOC01, VempalaWang04, AchlioptasMcSherryCOLT05, BrubakerVempalaFOCS08} provides rigorous guarantees  under separability conditions: roughly speaking, it is assumed that the constituent Gaussians have variation distance bounded away from $0$ (indeed, in some cases, distance exponentially close to $1$). This line of work was recently settled by a triplet of breakthrough results~\cite{KalaiMV10,MoitraValiantFOCS10,BelkinSinhaFOCS10}, establishing the polynomial solvability of the problem under minimal separability conditions for the parameters to be recoverable in the first place: for any $\ve>0$, polynomial in $n$ and $1/\ve$ samples from a mixture of $n$-dimensional Gaussians suffice to recover the parameters of the mixture in ${\rm poly}(n,1/\ve)$ time. 

While these results settle the polynomial solvability of the problem, they serve more as a proof of concept in that their dependence on $1/\ve$ is quite expensive.\footnote{For example, the single-dimensional algorithm in the heart of~\cite{KalaiMV10} has sample and time complexity of~$\Theta(1/\ve^{300})$ and~$\Omega(1/\ve^{1377})$ respectively (even though the authors most certainly did not intend to optimize their constants).} Indeed, even for mixtures of two single-dimensional Gaussians, a practically efficient algorithm for this problem is unknown.

\smallskip A weaker goal for the learner of Gaussian mixtures is this: given samples from an unknown mixture, find {\em any} mixture that is close to the unknown one, for some notion of {closeness}. This PAC-style version of the problem~\cite{KearnsMRRSSellieSTOC94} was pursued by Feldman et al.~\cite{FeldmanOS06} who obtained efficient learning algorithms for mixtures of $n$-dimensional, axis-aligned Gaussians. Given ${\rm poly}(n, 1/\ve, L)$ samples from such mixture, their algorithm constructs a mixture whose KL divergence  to the sampled one is at most $\ve$. Unfortunately, the sample and time complexity of their algorithm depends polynomially on a (priorly known) bound $L$, determining the range of the means and variances of the constituent Gaussians in every dimension.\footnote{In particular, it is assumed that every constituent Gaussian in every dimension has mean $\mu \in [-\mu_{\max}, \mu_{\max}]$ and variance $\sigma^2 \in [\sigma_{\min}^2, \sigma_{\max}^2]$ where $\mu_{\max} \sigma_{\max}/\sigma_{\min} \le L$.} In particular, the algorithm has pseudo-polynomial dependence on $L$ where there shouldn't be any dependence on $L$ at all~\cite{KalaiMV10,MoitraValiantFOCS10,BelkinSinhaFOCS10}.

\smallskip Finally, yet a weaker goal for the learner would be to construct {\em any distribution} that is close to the unknown mixture. In this {\em non-proper} version of the problem the learner is not restricted to output a Gaussian mixture, but can output any (representation of a) distribution that is close to the unknown mixture. For this problem, recent results of Chan et al.~\cite{ChanDSS13b} provide algorithms for single-dimensional mixtures, whose sample complexity has near-optimal dependence on~$1/\ve$. Namely, given $\tilde{O}(1/\ve^2)$ samples from a single-dimensional mixture, they construct a piecewise polynomial distribution that is $\ve$-close in total variation distance.

\smallskip Inspired by this recent progress on non-properly learning single-dimensional mixtures, our goal in this paper is to provide sample-optimal algorithms that {\em properly learn}. We obtain such algorithms for mixtures of two single-dimensional Gaussians. Namely,

\begin{theorem}\label{thm:main theorem}
For all $\ve, \delta > 0$, given $ \tilde{O}(\log(1/\delta)/\ve^2)$ independent samples from an arbitrary mixture~$F$ of two univariate Gaussians we can compute in time $\tilde O(\log^3(1/\delta)/\ve^{5})$ a mixture~$F'$ such that $\dtv(F,F') \le \ve$ with probability at least $1-\delta$. 
The expected running time of this algorithm is $\tilde O(\log^2(1/\delta)/\ve^{5})$.
\end{theorem}

We note that learning a univariate mixture often lies at the heart of learning multivariate mixtures \cite{KalaiMV10,MoitraValiantFOCS10}, so it is important to understand this fundamental case.

\paragraph{Discussion.} Note that our algorithm makes no separability assumptions about the constituent Gaussians of the unknown mixture, nor does it require or depend on a bound on the mixture's parameters. Also, because the mixture is single-dimensional it is not amenable to the techniques of~\cite{Hsu2013}. Moreover, it is easy to see that our sample complexity is optimal up to logarithmic factors. Indeed, a Gaussian mixture can trivially simulate a Bernoulli distribution as follows. Let $Z$ be a Bernoulli random variable that is $0$ with probability $1-p$ and $1$ with probability $p$. Clearly, $Z$ can be viewed as a mixture of two Gaussian random variables of $0$ variance, which have means $0$ and $1$ and are mixed with probabilities $1-p$ and $p$ respectively. It is known that $1/\ve^2$ samples are needed to properly learn a Bernoulli distribution, hence this lower bound immediately carries over to Gaussian mixtures. 

\paragraph{Approach.} Our algorithm is intuitively quite simple, although some care is required to make the ideas work. First, we can guess the mixing weight up to additive error $O(\ve)$, and proceed with our algorithm pretending that our guess is correct. Every guess will result in a collection of candidate distributions, and the final step of our algorithm is a tournament that will select, from among all candidate distributions produced in the course of our algorithm, a distribution that is $\ve$-close to the unknown mixture, if such a distribution exists. To do this we will make use of the following theorem which is our second main contribution in this paper. (See Theorem~\ref{thm:tournament theorem} for a precise statement.)

\medskip {\bf Informal Theorem~\ref{thm:tournament theorem}.}~~{\em There exists an algorithm {\tt FastTournament} that takes as input sample access to an unknown distribution $X$ and a collection of candidate hypothesis distributions $H_1,\ldots,H_N$, as well as an accuracy parameter $\ve >0$, and has the following behavior: if there exists some distribution among $H_1,\ldots,H_N$ that is $\ve$-close to $X$, then the distribution output by the algorithm is $O(\ve)$-close to $X$. Moreover, the number of samples drawn by the algorithm from each of the distributions is $O(\log N /\ve^2)$ and the running time is $O(N \log N/\ve^2)$.}

\medskip Devising a hypothesis selection algorithm with the performance guarantees of Theorem~\ref{thm:tournament theorem} requires strengthening the Scheff\'e-estimate-based approach described in Chapter~6 of~\cite{DL:01} (see \cite{Yatracos85,DL96,DL97}, as well as the recent papers of Daskalakis et al.~\cite{DaskalakisDiakonikolasServedioSTOC12} and Chan et al.~\cite{ChanDSS13a}) to continuous distributions whose crossings are difficult to compute exactly as well as to sample-only access to all involved distributions, but most importantly improving the running time to almost linear in the number  $N$ of candidate distributions. Further comparison of our new hypothesis selection algorithm to related work is provided in Section~\ref{sec:tournament}. It is also worth noting that the tournament based approach of~\cite{FeldmanOS06} cannot be used for our purposes in this paper as it would require a priorly known bound on the mixture's parameters and would depend pseudopolynomially on this bound.

\medskip Tuning the number of samples according to the guessed mixing weight, we proceed to draw samples from the unknown mixture, expecting that some of these samples will fall sufficiently close to the means of the constituent Gaussians, where the closeness will depend on the number of samples drawn as well as the unknown variances. We guess which sample falls close to the mean of the constituent Gaussian that has the smaller value of $\s/w$ (standard deviation to mixing weight ratio), which gives us the second parameter of the mixture. To pin down the variance of this Gaussian, we implement a natural idea. Intuitively, if we draw samples from the mixture, we expect that the constituent Gaussian with the smallest $\sigma/w$ will determine the smallest distance among the samples. Pursuing this idea we produce a collection of variance candidates, one of which truly corresponds to the variance of this Gaussian, giving us a third parameter.

At this point, we have a complete description of one of the component Gaussians. 
If we could remove this component from the mixture, we would be left with the remaining unknown Gaussian.
Our approach is to generate an empirical distribution of the mixture and ``subtract out'' the component that we already know, giving us an approximation to the unknown Gaussian. For the purposes of estimating the two parameters of this unknown Gaussian, we observe that the most traditional estimates of {location and scale} are unreliable, since the error in our approximation may cause probability mass to be shifted to arbitrary locations. Instead, we use robust statistics to obtain approximations to these two parameters.

The empirical distribution of the mixture is generated using the Dvoretzky-Kiefer-Wolfowitz (DKW) inequality \cite{dvoretzky1956}.
With $O(1/\ve^2)$ samples from an \emph{arbitrary} distribution, this algorithm generates an $\ve$-approximation to the distribution (with respect to the Kolmogorov metric).
Since this result applies to arbitrary distributions, it generates a hypothesis that is weak, in some senses, including the choice of distance metric.
In particular, the hypothesis distribution output by the DKW inequality is discrete, resulting in a total variation distance of $1$ from a mixture of Gaussians (or any other continuous distribution), regardless of the accuracy parameter $\ve$.
Thus, we consider it to be interesting that such a weak hypothesis can be used as a tool to generate a stronger, proper hypothesis.
We note that the Kolmogorov distance metric is not special here - an approximation with respect to other reasonable distance metrics may be substituted in, as long as the description of the hypothesis is efficiently manipulable in the appropriate ways.

We show that, for any target total variation distance $\ve$, the number of samples required to execute the steps outlined above in order to produce a collection of candidate hypotheses one of which is $\ve$-close to the unknown mixture, as well as to run the tournament to select from among the candidate distributions are $\tilde{O}(1/\ve^2)$. The running time is $\tilde{O}(1/\ve^5)$.

\paragraph{Comparison to Prior Work on Learning Gaussian Mixtures.}

In comparison to the recent breakthrough results~\cite{KalaiMV10,MoitraValiantFOCS10,BelkinSinhaFOCS10}, our algorithm has near-optimal sample complexity and much milder running time, where  these results have quite expensive dependence of both their sample and time complexity on the accuracy $\ve$, even for single-dimensional mixtures.\footnote{For example, compared to~\cite{KalaiMV10} we improve by a factor of at least $150$ the exponent of both the sample and time complexity.} On the other hand, our algorithm has weaker guarantees in that we properly learn but don't do parameter estimation. In comparison to~\cite{FeldmanOS06}, our algorithm requires no bounds on the parameters of the constituent Gaussians and exhibits no pseudo-polynomial dependence of the sample and time complexity on such bounds. On the other hand, we learn with respect to the total variation distance rather than the KL divergence. Finally, compared to~\cite{ChanDSS13a,ChanDSS13b}, we properly learn while they non-properly learn and we both have near-optimal sample complexity.

Recently and independently, Acharya et al. \cite{AJOS14a} have also provided algorithms for properly learning spherical Gaussian mixtures.
Their primary focus is on the high dimensional case, aiming at a near-linear sample dependence on the dimension. Our focus is instead on optimizing the dependence of the sample and time complexity on $\ve$ in the one-dimensional case. 

In fact, \cite{AJOS14a}  also study mixtures of $k$ Gaussians in one dimension, providing a proper learning algorithm  with near-optimal sample complexity of $\tilde O\left(k/\ve^2\right)$ and running time $\tilde O_k\left(1/\ve^{3k+1}\right)$. Specializing to two single-dimensional Gaussians ($k=2$), their algorithm has near-optimal sample complexity, like ours, but is slower by a factor of $O(1/\ve^2)$ than ours.
We also remark that, through a combination of techniques from our paper and theirs, a proper learning algorithm for mixtures of $k$ Gaussians can be obtained, with near-optimal sample complexity of $\tilde O\left(k/\ve^2\right)$ and running time $\tilde O_k\left(1/\ve^{3k-1}\right)$, improving by a factor of $O(1/\ve^2)$ the running time of their $k$-Gaussian algorithm. Roughly, this algorithm creates candidate distributions in which the parameters of the first $k-1$ components are generated using methods from \cite{AJOS14a}, and the parameters of the final component are determined using our robust statistical techniques, in which we ``subtract out'' the first $k-1$ components and robustly estimate the mean and variance of the remainder.

%% file: preliminaries.tex
\section{Preliminaries}
\label{sec:prelims}
Let $\mathcal{N}(\m,\s^2)$ represent the univariate normal distribution, with mean $\m \in \mathbb{R}$ and variance $\s^2 \in \mathbb{R}$, with density function 
$$\mathcal{N}(\m,\s^2,x) = \frac{1}{\s\sqrt{2\p}}e^{-\frac{(x-\m)^2}{2\s^2}}.$$

The univariate half-normal distribution with parameter $\s^2$ is the distribution of $|Y|$ where $Y$ is distributed according to $\mathcal{N}(0,\s^2)$. The CDF of the half-normal distribution is
$$F(\s, x) = \erf\left(\frac{x}{\s\sqrt{2}}\right),$$
where $\erf(x)$ is the error function, defined as
$$\erf(x) = \frac{2}{\sqrt{\p}}\int_0^x e^{-t^2}\, \mathrm{d} t.$$
We also make use of the complement of the error function, $\erfc(x)$, defined as $\erfc(x) = 1 - \erf(x)$.

A Gaussian mixture model (GMM) of distributions $\mathcal{N}_1(\m_1,\s_1^2), \dots, \mathcal{N}_n(\m_n,\s_n^2)$ has PDF $$f(x) = \sum_{i = 1}^n w_i \mathcal{N}(\m_i,\s_i^2,x),$$
where $\sum_i w_i = 1$.
These $w_i$ are referred to as the mixing weights.
Drawing a sample from a GMM can be visualized as the following process: select a single Gaussian, where the probability of selecting a Gaussian is equal to its mixing weight, and draw a sample from that Gaussian.
In this paper, we consider mixtures of two Gaussians, so $w_2 = 1 - w_1$.
We will interchangeably use $w$ and $1-w$ in place of $w_1$ and $w_2$.

The total variation distance between two probability measures $P$ and $Q$ on a $\s$-algebra $F$ is defined by
$$\dtv(P,Q) = \sup_{A \in F} |P(A) - Q(A)| = \frac12\|P - Q\|_1.$$

For simplicity in the exposition of our algorithm, we make the standard assumption (see, e.g.,~\cite{FeldmanOS06,KalaiMV10}) of infinite precision real arithmetic. In particular, the samples we draw from a mixture of Gaussians are real numbers, and we can do exact computations on real numbers, e.g., we can exactly evaluate the PDF of a Gaussian distribution on a real number.

%% file: tvbounds.tex
\subsection{Bounds on Total Variation Distance for GMMs}
We recall a result from \cite{DaskalakisDOST13}:
\begin{proposition}[Proposition B.4 of \cite{DaskalakisDOST13}]
  \label{prop:ddodtv}
  Let $\m_1, \m_2 \in \mathbb{R}$ and $0 \leq \s_1 \leq \s_2$.
  Then $$\dtv(\mathcal{N}(\m_1,\s_1^2),\mathcal{N}(\m_2, \s_2^2)) \leq \frac{1}{2}\left(\frac{|\m_1 - \m_2|}{\s_1} + \frac{\s_2^2 - \s_1^2}{\s_1^2}\right).$$
\end{proposition}

The following proposition, whose proof is deferred to the appendix, provides a bound on the total variation distance between two GMMs in terms of the distance between the constituent Gaussians. 
\begin{proposition}
  \label{prop:mixtv}
  Suppose we have two GMMs $X$ and $Y$, with PDFs $w\mathcal{N}_1 + (1-w)\mathcal{N}_2$ and $\hat w\mathcal{\hat N}_1 + (1 - \hat w)\mathcal{\hat N}_2$ respectively.
Then $\dtv(X,Y) \leq |w - \hat w| + w\dtv(\mathcal{N}_1, \mathcal{\hat N}_1) + (1-w)\dtv(\mathcal{N}_2, \mathcal{\hat N}_2)$.
\end{proposition}

Combining these propositions, we obtain the following lemma:
\begin{lemma}
  \label{lem:paramest}
  Let $X$ and $Y$ by two GMMs with PDFs $w_1\mathcal{N}_1 + w_2\mathcal{N}_2$ and $\hat w_1\mathcal{\hat N}_1 + \hat w_2\mathcal{\hat N}_2$ respectively, where $|w_i - \hat w_i| \leq O(\ve)$, $|\m_i - \hat \m_i| \leq O(\frac{\ve}{w_i})\s_i \leq O(\ve) \s_i$, $|\s_i - \hat \s_i| \leq O(\frac{\ve}{w_i}) \s_i \leq O(\ve)\s_i$, for all $i$ such that $w_i \geq \frac{\ve}{25}$.
  Then $\dtv(X,Y) \leq \ve$.
\end{lemma}

%% file: kolmogorov.tex
\subsection{Kolmogorov Distance}
In addition to total variation distance, we will also use the Kolmogorov distance metric.
\begin{definition}
  The \emph{Kolmogorov distance} between two probability measures with CDFs $F_X$ and $F_Y$ is
  $\dk(F_X,F_Y) = \sup_{x \in \mathbb{R}} |F_X(x) - F_Y(x)|$.
\end{definition}
We will also use this metric to compare general functions, which may not necessarily be valid CDFs.

We have the following fact, stating that total variation distance upper bounds Kolmogorov distance \cite{probmetricsreview}.
\begin{fact}
  \label{fact:dkdtv}
  $\dk(F_X,F_Y) \leq \dtv(f_X,f_Y)$
\end{fact}

Fortunately, it is fairly easy to learn with respect to the Kolmogorov distance, due to the Dvoretzky-Kiefer-Wolfowitz (DKW) inequality \cite{dvoretzky1956}.
\begin{theorem}{(\cite{dvoretzky1956},\cite{massart1990})}
  \label{thm:dkw}
  Suppose we have $n$ IID samples $X_1, \dots X_n$ from a probability distribution with CDF $F$.
  Let $F_n(x) = \frac{1}{n}\sum_{i=1}^n \mathbf{1}_{\{X_i \leq x\}}$ be the empirical CDF.
  Then $\Pr[\dk(F,F_n) \geq \ve] \leq 2e^{-2n\ve^2}$.
  In particular, if $n = \Omega((1/\ve^2) \cdot \log(1/\d))$, then $\Pr[\dk(F,F_n) \geq \ve] \leq \d$.
\end{theorem}

\subsection{Representing and Manipulating CDFs}
We will need to be able to efficiently represent and query the CDF of probability distributions we construct.
This will be done using a data structure we denote the \emph{$n$-interval partition} representation of a distribution.
This allows us to represent a discrete random variable $X$ over a support of size $ \leq n$.
Construction takes $\tilde O(n)$ time, and at the cost of $O(\log n)$ time per operation, we can compute $F_X^{-1}(x)$ for $x \in [0,1]$.
Full details will be provided in Appendix \ref{app:intervalpartition}.

Using this construction and Theorem $\ref{thm:dkw}$, we can derive the following proposition:
\begin{proposition}
  \label{prop:kolcons}
  Suppose we have $n = \Theta(\frac{1}{\ve^2}\cdot \log{\frac{1}{\d}})$ IID samples from a random variable $X$.
  In $\tilde O\left(\frac{1}{\ve^2} \cdot \log{\frac{1}{\d}}\right)$ time, we can construct a data structure which will
  allow us to convert independent samples from the uniform distribution over $[0,1]$ to independent samples from
  a random variable $\hat X$, such that $\dk\left(F_X,F_{\hat X}\right) \leq \ve$ with probability $1 - \d$.
\end{proposition}

Over the course of our algorithm, it will be natural to subtract out a component of a distribution.

\begin{lemma}
  \label{lem:mono}
  Suppose we have access to the $n$-interval partition representation of a CDF $F$, and that there exists a weight $w$ and CDFs $G$ and $H$ such that $\dk\left(H, \frac{F - wG}{1-w}\right) \leq \ve$.
  Given $w$ and $G$, we can compute the $n$-interval partition representation of a distribution $\hat H$ such that $\dk(H,\hat H) \leq \ve$ in $O(n)$ time.
\end{lemma}

A proof and full details are provided in Appendix \ref{app:intervalpartition}.

%% file: robust.tex
\subsection{Robust Statistics}
We use two well known robust statistics, the median and the interquartile range.
These are suited to our application for two purposes.
First, they are easy to compute with the $n$-interval partition representation of a distribution.
Each requires a constant number of queries of the CDF at particular values, and the cost of each query is $O(\log n)$.
Second, they are robust to small modifications with respect to most metrics on probability distributions.
In particular, we will demonstrate their robustness on Gaussians when considering distance with respect to the Kolmogorov metric.

\begin{lemma}
  \label{lem:med}
  Let $\hat F$ be a distribution such that $\dk(\mathcal{N}(\mu,\s^2),\hat F) \leq \ve$, where $\ve < \frac18$.
  Then $med(\hat F) \triangleq \hat F^{-1}(\frac{1}{2}) \in [\m - 2\sqrt{2}\ve\s, \m + 2\sqrt{2}\ve\s]$.
\end{lemma}
\begin{lemma}
  \label{lem:iqr}
  Let $\hat F$ be a distribution such that $\dk(\mathcal{N}(\mu,\s^2),\hat F) \leq \ve$, where $\ve < \frac18$.
  Then $\frac{IQR(\hat F)}{2\sqrt{2}\erf^{-1}(\frac12)} \triangleq \frac{\hat F^{-1}(\frac{3}{4}) - \hat F^{-1}(\frac{1}{4})}{2\sqrt{2}\erf^{-1}(\frac12)} \in \left[\s - \frac{5}{2\erf^{-1}(\frac12)}\ve\s, \s + \frac{7}{{2\erf^{-1}(\frac12)}}\ve\s\right]$.
\end{lemma}

The proofs are deferred to Appendix \ref{app:med}.

%% file: outline.tex
\subsection{Outline of the Algorithm}
We can decompose our algorithm into two components: generating a collection of candidate distributions containing at least one candidate with low statistical distance to the unknown distribution (Theorem~\ref{thm:list of candidates}), and identifying such a candidate from this collection (Theorem~\ref{thm:tournament theorem}).

\smallskip {\bf Generation of Candidate Distributions:} In Section \ref{sec:candidates}, we deal with generation of candidate distributions.
A \emph{candidate distribution} is described by the parameter set $(\hat w, \hat \m_1, \hat \s_1, \hat \m_2, \hat \s_2)$, which corresponds to the GMM with PDF 
$f(x) = \hat w\mathcal{N}(\hat \m_1, \hat \s_1^2,x) + (1-\hat w)\mathcal{N}(\hat \m_2, \hat \s_2^2,x)$.
As suggested by Lemma \ref{lem:paramest}, if we have a candidate distribution with sufficently accurate parameters, it will have low statistical distance to the unknown distribution.
Our first goal will be to generate a collection of candidates that contains at least one such candidate.
Since the time complexity of our algorithm depends on the size of this collection, we wish to keep it to a minimum.

At a high level, we sequentially generate candidates for each parameter.
In particular, we start by generating candidates for the mixing weight.
While most of these will be inaccurate, we will guarantee to produce at least one appropriately accurate candidate $\hat w^*$.
Then, for each candidate mixing weight, we will generate candidates for the mean of one of the Gaussians.
We will guarantee that, out of the candidate means we generated for $\hat w^*$, it is likely that at least one candidate $\hat \m_1^*$ will be sufficiently close to the true mean for this component.
The candidate means that were generated for other mixing weights have no such guarantee.
We use a similar sequential approach to generate candidates for the variance of this component.
Once we have a description of the first component, we simulate the process of subtracting it from the mixture, thus giving us a single Gaussian, whose parameters we can learn.
We can not immediately identify which candidates have inaccurate parameters, and they serve only to inflate the size of our collection.

At a lower level, our algorithm starts by generating candidates for the mixing weight followed by generating candidates for the mean of the component with the smaller value of $\frac{\s_i}{w_i}$.
Note that we do not know which of the two Gaussians this is.
The solution is to branch our algorithm, where each branch assumes a correspondence to a different Gaussian.
One of the two branches is guaranteed to be correct, and it will only double the number of candidate distributions.
We observe that if we take $n$ samples from a single Gaussian, it is likely that there will exist a sample at distance $O(\frac{\s}{n})$ from its mean.
Thus, if we take $\Theta(\frac{1}{ w_i \ve})$ samples from the mixture, one of them will be sufficiently close to the mean of the corresponding Gaussian. Exploiting this observation we obtain candidates for the mixing weight and the first mean as summarized by Lemma~\ref{lem:means}.

Next, we generate candidates for the variance of this Gaussian.
Our specific approach is based on the observation that given $n$ samples from a single Gaussian, the minimum distance of a sample to the mean will likely be $\Theta(\frac{\s}{n})$.
In the mixture, this property will still hold for the Gaussian with the smaller $\frac{\s_i}{w_i}$, so we extract this statistic and use a grid around it to generate sufficiently accurate candidates for $\s_i$. This is Lemma~\ref{lem:dev1}.

At this point, we have a complete description of one of the component Gaussians.
Also, we can generate an empirical distribution of the mixture, which gives an adequate approximation to the true distribution.
Given these two pieces, we update the empirical distribution by removing probability mass contributed by the known component.
When done carefully, we end up with an approximate description of the distribution of the unknown component.
At this point, we extract the median and the interquartile range (IQR) of the resulting distribution.
These statistics are robust, so they can tolerate error in our approximation.
Finally, the median and IQR allow us to derive the last mean and variance of our distribution.
This is Lemma~\ref{lem:lastparams}.

Putting everything together, we obtain the following result whose proof is in Section~\ref{sec:proof of theorem list of candidates}.
\begin{theorem}
  \label{thm:list of candidates}
  For all $\ve, \d > 0$, given $\log(1/\d) \cdot  O(1/\ve^2)$ independent samples from an arbitrary mixture $F$ of two univariate Gaussians, we can generate a collection of $\log(1/\d) \cdot \tilde O(1/\ve^3)$ candidate mixtures of two univariate Gaussians, containing at least one candidate $F'$ such that $\dtv(F, F') < \ve$ with probability at least $1 - \d$.
\end{theorem}

\smallskip {\bf Candidate Selection:} In view of Theorem~\ref{thm:list of candidates}, to prove our main result it suffices to select from among the candidate mixtures some mixture that is close to the unknown mixture. In Section~\ref{sec:tournament}, we describe a tournament-based algorithm (Theorem~\ref{thm:tournament theorem}) for identifying a candidate which has low statistical distance to the unknown mixture, concluding the proof of Theorem~\ref{thm:main theorem}. See our discussion in Section~\ref{sec:tournament} for the challenges arising in obtaining a tournament-based algorithm for continuous distributions whose crossings are difficult to compute exactly, as well as in speeding up the tournament's running time to almost linear in the number of candidates. 

%% file: mixing_weight.tex
\section{Generating Candidate Distributions}\label{sec:candidate generation}

By Proposition \ref{prop:mixtv}, if one of the Gaussians of a mixture has a negligible mixing weight, it has a negligible impact on the mixture's statistical distance to the unknown mixture.
Hence, the candidate means and variances of this Gaussian are irrelevant. This is fortunate, since if $\min{(w, 1-w)} << \ve$ and we only draw $O(1/\ve^2)$ samples from the unknown mixture, as we are planning to do, we have no hope of seeing a sufficient number of samples from the low-weight Gaussian to perform accurate statistical~tests~for~it. So for this section we will assume that $\min{(w, 1-w)} \ge \Omega(\ve)$ and we will deal with the other case separately.

\label{sec:candidates}
\subsection{Generating Mixing Weight Candidates}\label{sec:mixing weight candidates}
The first step is to generate candidates for the mixing weight.
We can obtain a collection of $O(\frac{1}{\ve})$ candidates containing some $\hat w^* \in [w - \ve, w + \ve]$ by simply taking the set $\{t\ve\mid t \in \left[\frac1{\ve}\right]\}$.

%% file: means.tex
\subsection{Generating Mean Candidates}
The next step is to generate candidates for the mean corresponding to the Gaussian with the smaller value of $\frac{\s_i}{w_i}$.
Note that, a priori, we do not know whether $i = 1$ or $i = 2$.
We try both cases, first generating candidates assuming they correspond to $\m_1$, and then repeating with $\m_2$.
This will multiply our total number of candidate distributions by a factor of $2$.
Without loss of essential generality,  assume for this section that $i = 1$. 

We want a collection of candidates containing $\hat \m_1^*$ such that 
$\m_1 - \ve\s_1 \leq \hat \m_1^* \leq \m_1 + \ve\s_1$. 
The following propositions are straightforward and proved in Appendix \ref{app:nearmean}.
\begin{proposition}
  \label{prop:nearmean}
  Fix $i\in\{1,2\}$. Given $\frac{20\sqrt{2}}{3w_i\ve}$ samples from a GMM, there will exist a sample $\hat \m_i^* \in \m_i \pm \ve\s_i$ with probability $\geq \frac{99}{100}$.
\end{proposition}

\begin{proposition}
  \label{prop:west}
  Fix $i\in\{1,2\}$.  Suppose $w_i - \ve \leq \hat w_i \leq w_i + \ve$, and $w_i \geq \ve$.
  Then $\frac{2}{\hat w_i} \geq \frac{1}{w_i}$.
\end{proposition}

We use these facts to design a simple algorithm: 
for each candidate $\hat w_1$ (from Section~\ref{sec:mixing weight candidates}), 
take $\frac{40\sqrt{2}}{3\hat w_1\ve}$ samples from the mixture and use each of them as a candidate for $\m_1$.

We now examine how many candidate pairs $(\hat{w},\hat{\mu}_1)$ we generated.
Naively, since $\hat{w}_i $ may be as small as $O(\ve)$, the candidates for the mean will multiply the size of our collection by $O\left(\frac{1}{\ve^2}\right)$.
However, we note that when $\hat{w}_i = \Omega(1)$, then the number of candidates for $\mu_i$ is actually $O\left(\frac{1}{\ve}\right)$.
We count the number of candidate triples $(\hat w, \hat \m_1)$, combining with previous results in the following:

\begin{lemma}
  \label{lem:means}
  Suppose we have sample access to a GMM with (unknown) parameters $(w,\m_1, \m_2, \s_1, \s_2)$.
  Then for any $\ve > 0$ and constants $c_w, c_m > 0$, using $O(\frac{1}{\ve^2})$ samples from the GMM, we can generate a collection of $O\left(\frac{\log{\ve^{-1}}}{\ve^2}\right)$ candidate pairs for $(w, \m_1)$.
  With probability $\geq \frac{99}{100}$, this will contain a pair $(\hat w^*, \hat \m_1^*)$ such that
 $\hat w^* \in w \pm O(\ve),$
 $\hat \m_1^* \in \m_1 \pm O(\ve)\s_1$.
\end{lemma}
\noindent The proof of the lemma is deferred to Appendix ~\ref{app:means}. It implies that we can generate $O\left(\frac1{\ve^2}\right)$ candidate triples, such that at least one pair simultaneously describes $w$ and $\m_1$ to the desired accuracy.

%% file: devs.tex
\subsection{Generating Candidates for a Single Variance}
In this section, we generate candidates for the variance corresponding to the Gaussian with the smaller value of $\frac{\s_i}{w_i}$.
We continue with our guess of whether $i = 1$ or $i = 2$ from the previous section.

Again, assume for this section that $i = 1$.
The basic idea is that we will find the closest point to $\hat \m_1$.
We use the following property (whose proof is deferred to Appendix \ref{sec:closestpointkol}) to establish a range for this distance, which we can then grid over.

We note that this lemma holds in scenarios more general than we consider here, including $k > 2$ and when samples are drawn from a distribution which is only close to a GMM, rather than exactly a GMM.
\begin{lemma}
  \label{lem:closeptgmmrob}
  Let $c_1$ and $c_2$ be constants as defined in Proposition \ref{thm:closept}, and $c_3 = \frac{c_1}{9\sqrt{2}c_2}$.
  Consider a mixture of $k$ Gaussians $f$, with components $\mathcal{N}(\m_1,\s_1^2), \dots, \mathcal{N}(\m_k, \s_k^2)$ and weights $w_1, \dots, w_k$, and let $j = \arg\min_i \frac{\s_i}{w_i}$.
  Suppose we have estimates for the weights and means for all $i \in [1,k]$:
  \begin{itemize}
    \item $\hat w_i$, such that $\frac12 \hat w_i \leq w_i \leq 2\hat w_i$
    \item $\hat \mu_i$, such that $|\hat \mu_i - \mu_i| \leq \frac{c_3}{2k}\s_j$
  \end{itemize}
  Now suppose we draw $n = \frac{9\sqrt{\p}c_2}{2\hat w_j}$ samples $X_1, \dots, X_n$ from a distribution $\hat f$, 
  where $\dk(f,\hat f) \leq \d = \frac{c_1}{2n} = \frac{c_1}{9\sqrt{\p}c_2}\hat w_j$.
  Then $\min_i |X_i - \hat \m_j| \in [\frac{c_3}{2k}\s_j, (\sqrt{2} + \frac{c_3}{2k})\s_j]$ with probability $\geq \frac{9}{10}$.
\end{lemma}

Summarizing what we have so far,

\begin{lemma}
  \label{lem:dev1}
  Suppose we have sample access to a GMM with parameters $(w, \m_1, \m_2, \s_1, \s_2)$, where $\frac{\s_1}{w} \leq \frac{\s_2}{1-w}$.
  Furthermore, we have estimates $\hat w^* \in w \pm O(\ve)$, $\hat \m_1^* \in \m_1 \pm O(\ve) \s_1$.
  Then for any $\ve > 0$, using $O(\frac{1}{\ve})$ samples from the GMM, we can generate a collection of $O\left(\frac1{\ve}\right)$ candidates for $\s_1$.
  With probability $\geq \frac{9}{10}$, this will contain a candidate $\hat \s_1^*$ such that $\hat \s_1^* \in (1 \pm O(\ve))\s_1$.
\end{lemma}

%% file: lastcomp.tex
\subsection{Learning the Last Component using Robust Statistics}
\label{sec:lastcomp}
At this point, our collection of candidates must contain a triple $(\hat w^*,\hat \m_1^*, \hat \s_1^*)$ 
which are sufficiently close to the correct parameters.
Intuitively, if we could remove this component from the mixture, we would be left with a distribution corresponding to a single Gaussian, which we could learn trivially.
We will formalize the notion of ``component subtraction,'' which will allow us to eliminate the known component and obtain a description of an approximation to the CDF for the remaining component.
Using classic robust statistics (the median and the interquartile range), we can then obtain approximations to the unknown mean and variance.
This has the advantage of a single additional candidate for these parameters, in comparison to $O(\frac{1}{\ve})$ candidates for the previous mean and variance.

Our first step will be to generate an approximation of the overall distribution.
We will do this only once, at the beginning of the entire algorithm.
Our approximation is with respect to the Kolmogorov distance.
Using the DKW inequality (Theorem \ref{thm:dkw}) and Proposition \ref{prop:kolcons}, we obtain a $O(\frac{1}{\ve^2})$-interval partition representation of $\hat H$ such that
$\dk(\hat H, H) \leq \ve$,  with probability $\geq 1 - \d$ using $O(\frac{1}{\ve^2} \cdot \log{\frac{1}{\d}})$ time and samples (where $H$ is the CDF of the GMM).

Next, for each candidate $(\hat w, \hat \m_1, \hat \s_1)$, we apply Lemma \ref{lem:mono} to obtain the $O(\frac{1}{\ve^2})$-interval partition of the distribution with the known component removed, i.e., using the notation of Lemma \ref{lem:mono}, let $H$ be the CDF of the GMM, $F$ is our DKW-based approximation to $H$, $w$ is the weight $\hat w$, and $G$ is $\mathcal{N}(\hat \m_1, \hat \s_1^2)$.
We note that this costs $O(\frac{1}{\ve^2})$ for each candidate triple, and since there are $\tilde O(\frac{1}{\ve^3})$ such triples, the total cost of such operations will be $\tilde O(\frac{1}{\ve^5})$.
However, since the tournament we will use for selection of a candidate will require $\tilde \Omega(\frac{1}{\ve^5})$ anyway, this does not affect the overall runtime of our algorithm.

The following proposition shows that, when our candidate triple is $(w^*,\m_1^*,\s_1^*)$, the distribution that we obtain after subtracting the known component out and rescaling is close to the unknown component.

\begin{proposition}
\label{prop:dksubtract}
Suppose there exists a mixture of two Gaussians $F = w\mathcal{N}(\m_1,\s_1^2) + (1-w)\mathcal{N}(\m_2,\s_2^2)$ where $O(\ve) \leq w \leq 1 - O(\ve)$,
and we have $\hat F$, such that $\dk(\hat F, F) \leq O(\ve)$, $\hat w^*$, such that $|\hat w^* - w| \leq O(\ve)$, $\hat \m_1^*$, such that $|\m_1^* - \m_1| \leq O(\ve)\s_1$, and $\hat \s_1^*$, such that $|\s_1^* - \s_1| \leq O(\ve)\s_1$.

Then $\dk\left(\mathcal{N}(\m_2,\s_2^2), \frac{\hat F - \hat w^* \mathcal{N}(\hat \m_1^*, \hat \s_1^{*2})}{1-\hat w^*}\right) \leq \frac{O(\ve)}{1 - w}$.
\end{proposition}

Since the resulting distribution is close to the correct one, we can use robust statistics (via Lemmas \ref{lem:med} and \ref{lem:iqr}) to recover the missing parameters.
We combine this with previous details into the following Lemma, whose proof is deferred to Appendix \ref{app:lastparams}.
\begin{lemma}
  \label{lem:lastparams}
  Suppose we have sample access to a GMM with parameters $(w, \m_1, \m_2, \s_1, \s_2)$, where $\frac{\s_1}{w} \leq \frac{\s_2}{1-w}$.
  Furthermore, we have estimates $\hat w^* \in w \pm O(\ve)$, $\hat \m_1^* \in \m_1 \pm O(\ve)\s_1$, $\hat \s_1^* \in (1 \pm O(\ve))\s_1$.
  Then for any $\ve > 0$, using $O(\frac{1}{\ve^2} \cdot \log\frac{1}{\d})$ samples from the GMM, with probability $\geq 1 - \d$, we can generate candidates $\hat \m_2^* \in \m_2 \pm O\left(\frac{\ve}{1-w}\right)\s_2$ and $\hat \s_2^* \in \left(1 \pm O\left(\frac{\ve}{1 -w}\right)\right)\s_2$.
\end{lemma}

%% file: puttingtogether.tex
\subsection{Putting It Together}\label{sec:proof of theorem list of candidates}

Theorem~\ref{thm:list of candidates} is obtained by combining the previous sections, and its proof is given in the appendix.

%% file: tournament.tex
\section{Quasilinear-Time Hypothesis Selection} \label{sec:tournament}

The goal of this section is to present a hypothesis selection algorithm, {\tt FastTournament}, which is given sample access to a target distribution $X$ and several hypotheses distributions $H_1,\ldots,H_N$, together with an accuracy parameter $\ve >0$,  and is supposed to select a hypothesis distribution from $\{H_1,\ldots,H_N\}$. The desired behavior is this: if at least one distribution in $\{H_1,\ldots,H_N\}$ is $\ve$-close to $X$ in total variation distance, we want that the hypothesis distribution selected by the algorithm is $O(\ve)$-close to $X$. We provide such an algorithm whose sample complexity is $O({1 \over \ve^2} \log N)$ and whose running time $O({1 \over \ve^2} N \log N)$, {\em i.e. quasi-linear in the number of hypotheses}, improving the running time of the state of the art (predominantly the Scheff\'e-estimate based algorithm in~\cite{DL:01}) quadratically. 

We develop our algorithm in full generality, assuming that we have sample access to the distributions of interest, and without making any assumptions about whether they are continuous or discrete, and whether their support is single- or multi-dimensional. All our algorithm needs is sample access to the distributions at hand, together with a way to compare the probability density/mass functions of the distributions, encapsulated in the following definition. In our definition, $H_i(x)$ is the probability mass at $x$ if $H_i$ is a discrete distribution, and the probability density at $x$ if $H_i$ is a continuous distribution. We assume that $H_1$ and $H_2$ are either both discrete or both continuous, and that, if they are continuous, they have a density function.

\begin{definition}
Let $H_1$ and $H_2$ be probability distributions over some set ${\cal D}$. A {\em PDF comparator for $H_1, H_2$} is an oracle that takes as input some $x \in {\cal D}$ and outputs $1$ if $H_1(x)>H_2(x)$, and $0$ otherwise. 
\end{definition}

Our hypothesis selection algorithm is summarized in the following statement:

\begin{theorem} \label{thm:tournament theorem}
There is an algorithm {\tt FastTournament}$(X, {\cal H},\ve,\delta)$, which is given sample access to some distribution $X$ and a collection of distributions ${\cal H}=\{H_1,\ldots,H_N\}$ over some set ${\cal D}$, access to a PDF comparator for every pair of distributions $H_i, H_j \in {\cal H}$, an accuracy parameter $\ve >0$, and a confidence parameter $\delta >0$.  The algorithm makes
{$O\left({\log {1/ \delta} \over \ve^2} \cdot \log N\right)$} draws from each of $X, H_1,\ldots,H_N$ and returns some $H \in {\cal H}$ or declares ``failure.''  If there is some $H^* \in {\cal H}$ such that $\dtv(H^*,X) \leq \ve$ then with probability at least $1-\delta$ the distribution $H$ that {\tt
FastTournament} returns satisfies $\dtv(H,X) \leq {512} \ve.$ The total number of operations of the algorithm is {$O\left( {\log{1 / \delta} \over \ve^2} \left(N \log N + \log^2 {1 \over \delta}\right) \right)$}.
Furthermore, the expected number of operations of the algorithm is {$O\left( {N\log{N /\d} \over \ve^2}\right)$}.
\end{theorem}

The proof of Theorem~\ref{thm:tournament theorem} is given in Section~\ref{sec:fast tournament}, while the preceding sections build the required machinery for the construction.
 
\begin{remark} A slight modification of our algorithm provided in Appendix~{\ref{sec: tournament appendix 2}} admits a worst-case running time of {$O\left( {\log{1 / \delta} \over \ve^2} \left(N \log N + \log^{1+\g} {1 \over \delta}\right) \right)$}, for any desired constant $\g > 0$, though the approximation guarantee is weakened based on the value of $\g$. See Corollary~\ref{corr:recursive tournament corollary} and its proof in Appendix~\ref{sec: tournament appendix 2}.
\end{remark}

 \paragraph{Comparison to Other Hypothesis Selection Methods:} The skeleton of the hypothesis selection algorithm of Theorem~\ref{thm:tournament theorem} as well as the improved one of Corollary~\ref{corr:recursive tournament corollary}, is having candidate distributions compete against each other in a tournament-like fashion. This approach is quite natural and has been commonly used in the literature; see e.g. Devroye and Lugosi (\cite{DL96,DL97} and Chapter~6 of \cite{DL:01}), Yatracos \cite{Yatracos85}, as well as the recent papers of~Daskalakis et al.~\cite{DaskalakisDiakonikolasServedioSTOC12} and Chan et al.~\cite{ChanDSS13a}. 
 The hypothesis selection algorithms in these works are significantly slower than ours, as their running times have quadratic dependence on the number $N$ of hypotheses, while our dependence is quasi-linear.
 Furthermore, our setting is more general than prior work, in that we only require sample access to the hypotheses and a PDF comparator.
 Previous algorithms required knowledge of (or ability to compute) the probability assigned by every pair of hypotheses to their Scheff\'e set---this is the subset of the support where one hypothesis has larger PMF/PDF than the other, which  is difficult to compute in general, even given explicit descriptions of the hypotheses.
 
 Recent independent work by Acharya et al.~\cite{AJOS14a,AJOS14b} provides a hypothesis selection algorithm, based on the Scheff\'{e} estimate  in Chapter 6 of \cite{DL:01}. Their  algorithm  performs a number of operations that is comparable to ours. In particular, the expected running time of their algorithm is also $O\left( {N\log{N /\d} \over \ve^2}\right)$, but our worst-case running time has better dependence on $\d$. Our algorithm is not based on the Scheff\'e estimate, using instead a specialized estimator provided in Lemma~\ref{lem:choosehypothesis}.   Their algorithm, described in terms of the Scheff\'{e} estimate, is not immediately applicable to sample-only access to the hypotheses, or to settings where the probabilities on Scheff\'e sets are difficult to compute.
 
\subsection{Choosing Between Two Hypotheses} \label{sec:choose hypothesis}

We start with an algorithm for choosing between two hypothesis distributions. This is an adaptation of a similar algorithm from~\cite{DaskalakisDiakonikolasServedioSTOC12} to continuous distributions and sample-only access. The proof of the following is in Appendix~\ref{sec: tournament appendix}.

\begin{lemma} \label{lem:choosehypothesis}
There is an algorithm {\tt ChooseHypothesis}$(X,{H}_1,{H}_2,\ve,\delta)$, which
is given sample access to distributions $X, H_1, H_2$ over some set $\cal D$, access to a PDF comparator for $H_1,H_2$, an accuracy
parameter $\ve >0$, and a confidence parameter $\delta>0.$  The algorithm draws $m=O(\log(1/\delta)/\ve^2)$
samples from each of $X, H_1$ and $H_2$, and either returns some $H \in \{H_1,H_2\}$ as the winner or declares a ``draw.'' The total number of operations of  the algorithm is $O(\log(1/\delta)/\ve^2)$. Additionally, the output satisfies the following properties:
\begin{enumerate}
\item  If $\dtv(X,H_1) \leq \ve$ but $\dtv(X, H_2)>8 \ve$, the probability that $H_1$ is not declared winner~is~$\le \delta$; \label{choose hypothesis:property 1}
\item If $\dtv(X,H_1) \leq \ve$ but $\dtv(X, H_2)>4 \ve$,  the probability that $H_2$ is declared winner is $\le \delta$;\label{choose hypothesis:property 2}
\item The analogous conclusions hold if we interchange $H_1$ and $H_2$ in Properties~\ref{choose hypothesis:property 1} and~\ref{choose hypothesis:property 2} above;
\item If $\dtv(H_1, H_2) \le 5 \ve$, the algorithm declares a ``draw'' with probability at least $1-\delta$.
\end{enumerate}
\end{lemma}

\subsection{The Slow Tournament}\label{sec:slow tournament}

We proceed with a hypothesis selection algorithm, {\tt SlowTournament}, which has the correct behavior, but whose running time is suboptimal. Again we proceed similarly to~\cite{DaskalakisDiakonikolasServedioSTOC12} making the approach robust to continuous distributions and sample-only access. {\tt SlowTournament} performs pairwise comparisons between all  hypotheses in ${\cal H}$, using the subroutine {\tt ChooseHypothesis} of Lemma~\ref{lem:choosehypothesis}, and outputs a hypothesis that never lost to (but potentially tied with) other hypotheses. The running time of the algorithm is quadratic in $|{\cal H}|$, as all pairs of hypotheses are compared. {\tt FastTournament}, described in Section~\ref{sec:fast tournament}, organizes the tournament in a more efficient manner, improving the running time to quasilinear. The proof of Lemma~\ref{thm:slow tournament theorem} can be found in Appendix~\ref{sec: tournament appendix}.

\begin{lemma} \label{thm:slow tournament theorem}
There is an algorithm {\tt SlowTournament}$(X, {\cal H},\ve,\delta)$, which is given sample access to some distribution $X$ and a collection of distributions ${\cal H}=\{H_1,\ldots,H_N\}$ over some set ${\cal D}$, access to a PDF comparator for every pair of distributions $H_i, H_j \in {\cal H}$, an accuracy parameter $\ve >0$, and a confidence parameter $\delta >0$.  The algorithm makes
$m=O(\log(N/\delta) /\ve^2)$ draws from each of $X, H_1,\ldots,H_N$ and returns some $H \in {\cal H}$ or declares ``failure.''  If there is some $H^* \in {\cal H}$ such that $\dtv(H^*,X) \leq \ve$ then with probability at least $1-\delta$ the distribution $H$ that {\tt
SlowTournament} returns satisfies $\dtv(H,X) \leq 8 \ve.$ The total number of operations of the algorithm is $O\left( N^2 \log(N/\delta) /\ve^2 \right)$.
\end{lemma}

\subsection{The Fast Tournament}\label{sec:fast tournament}

We prove our main result of this section, providing a quasi-linear time algorithm for selecting from a collection of hypothesis distributions ${\cal H}$ one that is close to a target distribution $X$, improving the running time of {\tt SlowTournament} from Lemma~\ref{thm:slow tournament theorem}. Intuitively, there are two cases to consider. Collection ${\cal H}$ is either dense or sparse in distributions that are close to $X$. In the former case, we show that we can sub-sample ${\cal H}$ before running ${\tt SlowTournament}$. In the latter case, we show how to set-up a two-phase tournament, whose first phase  eliminates all but a sub linear number of hypotheses, and whose second phase runs {\tt SlowTournament} on the surviving hypotheses. Depending on the density of ${\cal H}$ in distributions that are close to the target distribution $X$, we show that one of the aforementioned strategies is guaranteed to output a distribution that is close to $X$.  As we do not know a priori the density of ${\cal H}$ in distributions that are close to $X$, and hence which of our two strategies will succeed in finding a distribution that is close to $X$, we use both strategies and run a tournament among their outputs, using {\tt SlowTournament} again.\\

\begin{prevproof}{Theorem}{thm:tournament theorem}
Let $p$ be the fraction of the elements of ${\cal H}$ that are $8\ve$-close to $X$. The value of $p$ is unknown to our algorithm. Regardless, we propose two strategies for selecting a distribution from ${\cal H}$, one of which is guaranteed to succeed whatever the value of $p$ is. We assume throughout this proof that $N$ is larger than a sufficiently large constant,  otherwise our claim follows directly from Lemma~\ref{thm:slow tournament theorem}.

\paragraph{S1:} Pick a random subset ${\cal H}' \subseteq {\cal H}$ of size $\lceil {3 \sqrt{N}} \rceil$, and run {\tt SlowTournament}$(X,{\cal H}', 8\ve, e^{-3})$ to select some distribution $\tilde{H} \in {\cal H}'$. We show the following in Appendix~\ref{sec: tournament appendix}.

\begin{claim} \label{claim1:tournament}
The number of samples drawn by {\rm S1} from each of the distributions in ${\cal H} \cup \{X\}$ is $O({1\over \ve^2} \log N)$, and the total number of operations is $O({1\over \ve^2} N \log N)$. Moreover, if  $p \in [{1 \over \sqrt{N}},1]$ and there is some distribution in ${\cal H}$ that is $\ve$-close to $X$, then the distribution $\tilde{H}$ output by {\rm S1} is $64 \ve$-close to $X$, with probability at least $9/10$.\end{claim}

\paragraph{S2:} There are two phases in this strategy: 
\begin{itemize}
\item {\bf Phase 1:} This phase proceeds in {$T = \lfloor \log_2 {\sqrt{N} \over 2} \rfloor$} iterations, $i_1,\ldots,i_T$. Iteration $i_\ell$ takes as input a subset ${\cal H}_{i_{\ell-1}} \subseteq {\cal H}$ (where ${\cal H}_{i_0} \equiv {\cal H}$), and produces some ${\cal H}_{i_{\ell}} \subset {\cal H}_{i_{\ell-1}}$, such that $|{\cal H}_{i_{\ell}}| = \left\lceil {|{\cal H}_{i_{\ell-1}}| \over 2} \right\rceil$,  as follows: randomly pair up the elements of ${\cal H}_{i_{\ell-1}}$ (possibly one element is left unpaired), and for every pair $(H_i, H_j)$ run {\tt ChooseHypothesis}$(X,{H}_i,{H}_j,\ve,1/3N)$. We do this with a small caveat: instead of drawing $O(\log(3N) /\ve^2)$  fresh samples (as required by Lemma~\ref{lem:choosehypothesis}) in every execution of {\tt ChooseHypothesis} (from whichever distributions are involved in that execution), we draw
$O(\log(3N) /\ve^2)$ samples from each of $X, H_1,\ldots,H_N$ once and for all, and reuse the same samples in all executions of {\tt ChooseHypothesis}.

\item {\bf Phase 2:} Given the collection ${\cal H}_{i_T}$ output by Phase 1, we run {\tt SlowTournament}$(X,{\cal H}_{i_T}, \ve, 1/4)$ to select some distribution $\hat{H} \in {\cal H}_{i_T}$. (We use fresh samples for the execution of {\tt SlowTournament}.)
\end{itemize}
We show the following in Appendix~\ref{sec: tournament appendix}.
\begin{claim} \label{claim2:tournament}
The number of samples drawn by {\rm S2} from each of the distributions in ${\cal H} \cup \{X\}$ is $O({1\over \ve^2} \log N)$, and the total number of operations is $O({1 \over \ve^2} N \log N)$. Moreover, if $p \in (0,{1 \over \sqrt{N}}]$ and there is some distribution in ${\cal H}$ that is $\ve$-close to $X$, then the distribution $\hat{H}$ output by {\rm S2} is $8 \ve$-close to $X$, with probability at least $1/4$.\end{claim}

Given strategies {\rm S1} and {\rm S2}, we first design an algorithm which has the stated worst-case number of operations.
The algorithm {\tt FastTournament}$_A$ works as follows:
\begin{enumerate}
\item Execute strategy {\rm S1} $k_1=\log_2 {2 \over \delta}$ times, with fresh samples each time. Let $\tilde{H}_1,\ldots,\tilde{H}_{k_1}$ be the distributions output by these executions.
\item Execute strategy {\rm S2} $k_2=\log_4 {2 \over \delta}$ times, with fresh samples each time. Let $\hat{H}_1,\ldots,\hat{H}_{k_2}$ be the distributions output by these executions.
\item Set ${\cal G} \equiv \{\tilde{H}_1,\ldots,\tilde{H}_{k_1}, \hat{H}_1,\ldots,\hat{H}_{k_2}\}$. Execute {\tt SlowTournament}$(X, {\cal G}, 64 \ve, \delta/2)$.
\end{enumerate}
\begin{claim}\label{claim:fast tournament a final claim}
{\tt FastTournament}$_A$ satisfies the properties described in the statement of Theorem~\ref{thm:tournament theorem}, except for the bound on the expected number of operations.
\end{claim}
\begin{prevproof}{Claim}{claim:fast tournament a final claim}
The bounds on the number of samples and operations follow immediately from our choice of $k_1, k_2$, Claims~\ref{claim1:tournament} and~\ref{claim2:tournament}, and Lemma~\ref{thm:slow tournament theorem}. Let us justify the correctness of the algorithm. Suppose that there is some distribution in ${\cal H}$ that is $\ve$-close to $X$. We distinguish two cases, depending on the fraction $p$ of distributions in ${\cal H}$ that are $\ve$-close to $X$:
\begin{itemize}
\item $p \in [{1 \over \sqrt{N}},1]$: In this case, each execution of {\rm S1} has probability at least $9/10$ of outputting a distribution that is $64 \ve$-close to $X$. So the probability that none of $\tilde{H}_1,\ldots,\tilde{H}_{k_1}$ is $64 \ve$-close to $X$ is at most $({1 \over 10})^{k_1} \le \delta/2$. Hence, with probability at least $1-{\delta/2}$, $\cal G$ contains a distribution that is $64\ve$-close to $X$. Conditioning on this, {\tt SlowTournament}$(X, {\cal G}, 64 \ve, \delta/2)$ will output a distribution that is $512\ve$-close to $X$ with probability at least $1-\delta/2$, by Lemma~\ref{thm:slow tournament theorem}. Hence, with overall probability at least $1-\delta$, the distribution output $\text{by {\tt FastTournament} is $512 \ve$-close~to~$X$.}$

\item $p \in (0,{1 \over \sqrt{N}}]$: This case is analyzed analogously. With probability at least $1-\delta/2$, at least one of $\hat{H}_1,\ldots,\hat{H}_{k_2}$ is $8\ve$-close to $X$ (by Claim~\ref{claim2:tournament}). Conditioning on this, {\tt SlowTournament}$(X, {\cal G}, 64 \ve, \delta/2)$ outputs a distribution that is $512\ve$-close to $X$, with probability at least $1-\delta/2$ (by Lemma~\ref{thm:slow tournament theorem}). So, with overall probability at least $1-\delta$, the distribution output {by {\tt FastTournament} is $512 \ve$-close~to~$X$.}
\end{itemize}
\end{prevproof}

We now describe an algorithm which has the stated expected number of operations.
The algorithm {\tt FastTournament}$_B$ works as follows:
\begin{enumerate}
\item Execute strategy {\rm S1}, let $\tilde{H}_1$ be the distribution output by this execution.
\item Execute strategy {\rm S2}, let $\tilde{H}_2$ be the distribution output by this execution.
\item Execute {\tt ChooseHypothesis}$(X,\tilde{H}_i,H,64\ve,\delta/N^3)$ for $i \in \{1,2\}$ and all $H \in \mathcal{H}$.
  If either $\tilde{H}_1$ or $\tilde{H}_2$ never loses, output that hypothesis.
  Otherwise, remove $\tilde{H}_1$ and $\tilde{H}_2$ from $\mathcal{H}$, and repeat the algorithm starting from step $1$, unless ${\cal H}$ is empty.
\end{enumerate}

\begin{claim}\label{claim:fast tournament b final claim}
{\tt FastTournament}$_B$ satisfies the properties described in the statement of Theorem~\ref{thm:tournament theorem}, except for the worst-case bound on the number of operations.
\end{claim}
\begin{prevproof}{Claim}{claim:fast tournament b final claim}
  We note that we will first draw $O(\log(N^3/\d)/\ve^2)$ from each of $X, H_1,\dots,H_N$ and use the same samples for every execution of {\tt ChooseHypothesis} to avoid blowing up the sample complexity.
  Using this fact, the sample complexity is as claimed.

  We now justify the correctness of the algorithm.
  Since we run {\tt ChooseHypothesis} on a given pair of hypotheses at most once, there are at most $N^2$ executions of this algorithm.
  Because each fails with probability at most $\frac{\d}{N^3}$, by the union bound, the probability that any execution of {\tt ChooseHypothesis} ever fails is at most $\d$, so all executions succeed with probability at least $1 - {\d \over N}$. Condition on this happening for the remainder of the proof of correctness. In Step 3 of our algorithm, we compare some $\tilde{H}$ with every other hypothesis. We analyze two cases:
  \begin{itemize}
    \item Suppose that $\dtv(X,\tilde{H}) \leq 64\ve$.
      By Lemma~\ref{lem:choosehypothesis}, $\tilde{H}$ will never lose, and will be output by {\tt FastTournament}$_B$.
    \item Suppose that $\dtv(X,\tilde{H}) > 512\ve$.
      Then by Lemma~\ref{lem:choosehypothesis}, $\tilde{H}$ will lose to any candidate $H'$ with $\dtv(X,H') \leq 64\ve$.
      We assumed there exists at least one hypothesis with this property in the beginning of the algorithm.
      Furthermore, by the previous case, if this hypothesis were selected by {\rm S1} or {\rm S2} at some prior step, the algorithm would have terminated; so in particular, if the algorithm is still running, this hypothesis could not have been removed from $\mathcal{H}$.
      Therefore, $\tilde{H}$ will lose at least once and will not be output by {\tt FastTournament}$_B$.
  \end{itemize}
  The correctness of our algorithm follows from the second case above. Indeed, if the algorithm outputs a distribution $\tilde{H}$, it must be the case that $\dtv(X,\tilde{H}) \leq 512\ve$. 
  Moreover, we will not run out of hypotheses before we output a distribution.
  Indeed, we only discard a hypothesis if it was selected by {\rm S1} or {\rm S2} and then lost at least once in Step 3. Furthermore, in the beginning of our algorithm there exists a distribution $H$ such that $\dtv(X,H) \leq 64\ve$. If ever selected  by {\rm S1} or {\rm S2}, $H$ will not lose to any distribution in Step 3, and the algorithm will output a distribution. If it is not selected by {\rm S1} or {\rm S2},  $H$ won't be removed from ${\cal H}$.

  We now reason about the expected running time of our algorithm. 
  First, consider the case when all executions of {\tt ChooseHypothesis} are successful, which happens with probability $\geq 1 - \frac{\d}{N}$.
  If either {\rm S1} or {\rm S2} outputs a distribution such that $\dtv(X,\tilde{H}) \leq 64\ve$, then by the first case above it will be output by {\tt FastTournament}$_B$.
  If this happened with probability at least $p$ independently in every iteration of our algorithm, then the number of iterations would be stochastically dominated by a geometric random variable with parameter $p$, so the expected number of rounds would be upper bounded by $\frac{1}{p}$.
  By Claims~\ref{claim1:tournament} and~\ref{claim2:tournament}, $p \geq \frac14$, so, when {\tt ChooseHypothesis} never fails, the expected number of rounds is at most $4$.
  Next, consider when at least one execution of {\tt ChooseHypothesis} fails, which happens with probability $\leq \frac{\d}{N}$.
  Since {\tt FastTournament}$_B$ removes at least one hypothesis in every round, there are at most $N$ rounds.
  Combining these two cases, the expected number of rounds is at most $(1 - \frac{\d}{N})4 + \frac{\d}{N}N \leq 5$.
  
  By Claims~\ref{claim1:tournament} and~\ref{claim2:tournament} and Lemma~\ref{lem:choosehypothesis}, each round requires $O(N\log N + N \log{N/\d})$ operations.
  Since the expected number of rounds is $O(1)$, we obtain the desired bound on the expected number of operations.
\end{prevproof}

In order to obtain all the guarantees of the theorem simultaneously, our algorithm {\tt FastTournament} will alternate between steps of {\tt FastTournament}$_A$ and {\tt FastTournament}$_B$, where both algorithms are given an error parameter equal to $\frac{\d}{2}$.
If either algorithm outputs a hypothesis, {\tt FastTournament} outputs it.
By union bound and Claims~\ref{claim:fast tournament a final claim} and \ref{claim:fast tournament b final claim}, both {\tt FastTournament}$_A$ and {\tt FastTournament}$_B$ will  be correct with probability at least $1 - \d$.
The worst-case running time is as desired by Claim~\ref{claim:fast tournament a final claim} and since interleaving between steps of the two tournaments will multiply the number of steps by a factor of at most $2$.
We have the expected running time similarly, by Claim~\ref{claim:fast tournament b final claim}.

\end{prevproof}

\section{Proof of Theorem~\ref{thm:main theorem}}

Theorem~\ref{thm:main theorem} is an immediate consequence of Theorems~\ref{thm:list of candidates} and~\ref{thm:tournament theorem}. Namely,  we run the algorithm of Theorem~\ref{thm:list of candidates} to produce a collection of Gaussian mixtures, one of which is within $\ve$ of the unknown mixture $F$. Then we use {\tt FastTournament} of Theorem~\ref{thm:tournament theorem} to select from among the candidates a mixture that is $O(\ve)$-close to $F$. For the execution of {\tt FastTournament}, we need a PDF comparator for all pairs of candidate mixtures in our collection. Given that these are described with their parameters, our PDF comparators evaluate the densities of two given mixtures at a challenge point $x$ and decide which one is largest. We also need sample access to our candidate mixtures. Given a parametric description $(w, \mu_1, \sigma_1^2, \mu_2, \sigma_2^2 )$ of a mixture, we can draw a sample from it as follows: first draw a uniform $[0,1]$ variable whose value compared to $w$ determines whether to sample from ${\cal N}(\mu_1, \sigma_1^2)$ or ${\cal N}(\mu_2, \sigma_2^2)$ in the second step; for the second step, use the Box-Muller transform~\cite{BoxM58} to obtain sample from either ${\cal N}(\mu_1, \sigma_1^2)$ or ${\cal N}(\mu_2, \sigma_2^2)$ as decided in the first step.

%% file: gridding.tex
\section{Gridding}
\label{app:grid}
We will encounter settings where we have bounds $L$ and $R$ on an unknown measure $X$ such that $L \leq X \leq R$, and wish to obtain an estimate $\hat X$ such that $(1 - \ve)X \leq \hat X \leq (1 + \ve)X$.
Gridding is a common technique to generate a list of candidates that is guaranteed to contain such an estimate.
\begin{fact}
  \label{fact:add}
  Candidates of the form $L + k\ve L$ define an additive grid with at most $\frac{1}{\ve}\left(\frac{R-L}{L}\right)$ candidates.
\end{fact}

\begin{fact}
  \label{fact:mult}
  Candidates of the form $L(1 + \ve)^k$ define a multiplicative grid with at most $\frac{1}{\log{(1 + \ve)}}\log{\left(\frac{R}{L}\right)}$ candidates.
\end{fact}

We also encounter scenarios where we require an additive estimate $X - \ve \leq \hat X \leq X + \ve$.

\begin{fact}
  \label{fact:absadd}
  Candidates of the form $L + k\ve$ define an absolute additive grid with $\frac{1}{\ve}\left(R-L\right)$ candidates.
\end{fact}

%% file: proofs.tex
\section{Proofs Omitted from Section~\ref{sec:prelims}}

\begin{prevproof}{Proposition}{prop:mixtv}
  We use $\dtv(P,Q)$ and $\frac12\|P - Q\|_1$ interchangeably in the cases where $P$ and $Q$ are not necessarily probability distributions.
  Let $\mathcal{N}_i = \mathcal{N}(\m_i, \s_i^2)$ and $ \mathcal{\hat N}_i = \mathcal{N}(\hat \m_i, \hat \s_i^2)$.
  By triangle inequality, 
  $$\dtv(\hat w\mathcal{\hat N}_1 + (1 - \hat w)\mathcal{\hat N}_2, w\mathcal{N}_1 + (1 - w)\mathcal{N}_2)  \\
  \leq \dtv(\hat w\mathcal{\hat N}_1, w\mathcal{N}_1) +  \dtv((1 - \hat w)\mathcal{\hat N}_2, (1 - w)\mathcal{N}_2) $$
  Inspecting the first term,
  $$\frac12\left\|w \mathcal{N}_1 - \hat w \mathcal{\hat N}_1\right\|_1 =
    \frac12\left\|w \mathcal{N}_1 - w \mathcal{\hat N}_1 + w \mathcal{\hat N}_1 - \hat w \mathcal{\hat N}_1\right\|_1 
  \leq w\dtv(\mathcal{N}_1,\mathcal{\hat N}_1) + \frac12\left|w - \hat w\right|, $$
  again using the triangle inequality.
  A symmetric statement holds for the other term, giving us the desired result.
\end{prevproof}

\begin{prevproof}{Lemma}{lem:med}
  \label{app:med}
  We will use $x$ to denote the median of our distribution, where $\hat F(x) = \frac12$.
  Since $\dk(F, \hat F) \leq \ve$, $F(x) \leq \frac12 + \ve$.
  Using the CDF of the normal distribution, we obtain
  $\frac12 + \frac12\erf\left(\frac{x - \m}{\sqrt{2\s^2}}\right) \leq \frac12 + \ve$.
  Rearranging, we get $x \leq \m + \sqrt{2}\erf^{-1}(2\ve)\s \leq \m + 2\sqrt{2}\ve\s$, where the inequality uses the Taylor series of $\erf^{-1}$ around $0$ and Taylor's Theorem.
  By symmetry of the Gaussian distribution, we can obtain the corresponding lower bound for $x$.
\end{prevproof}

\begin{prevproof}{Lemma}{lem:iqr}
  \label{app:iqr}
  First, we show that 
  $$F^{-1}\left(\frac34\right) \in 
  \left[\m + \sqrt{2}\erf^{-1}\left(\frac12\right)\s - \frac{5\sqrt{2}}{2}\s\ve,
   \m + \sqrt{2}\erf^{-1}\left(\frac12\right)\s + \frac{7\sqrt{2}}{2}\s\ve\right].$$

  Let $x = F^{-1}\left(\frac34\right)$.
  Since $\dk(F, \hat F) \leq \ve$, $F(x) \leq \frac34 + \ve$.
  Using the CDF of the normal distribution, we obtain
  $\frac12 + \frac12\erf\left(\frac{x - \m}{\sqrt{2\s^2}}\right) \leq \frac34 + \ve$.
  Rearranging, we get $x \leq \m + \sqrt{2}\erf^{-1}\left(\frac12 + 2\ve\right)\s \leq \m + \sqrt{2}\erf^{-1}\left(\frac12\right)\s + \frac{7\sqrt{2}}2\ve\s$, where the inequality uses the Taylor series of $\erf^{-1}$ around $\frac12$ and Taylor's Theorem.
  A similar approach gives the desired lower bound.

  By symmetry, we can obtain the bounds
  $$F^{-1}\left(\frac14\right) \in  
  \left[\m - \sqrt{2}\erf^{-1}\left(\frac12\right)\s - \frac{7\sqrt{2}}{2}\s\ve,
   \m - \sqrt{2}\erf^{-1}\left(\frac12\right)\s + \frac{5\sqrt{2}}{2}\s\ve\right].$$
  Combining this with the previous bounds and rescaling, we obtain the lemma statement.\end{prevproof}
\section{Proofs Omitted from Section~\ref{sec:candidate generation}}

\begin{prevproof}{Proposition}{prop:nearmean}
\label{app:nearmean}
  The probability that a sample is from $\mathcal{N}_i$ is $w_i$.
  Using the CDF of the half-normal distribution, given that a sample is from $\mathcal{N}_i$, the probability that it is at a distance $\leq \ve\s_i$ from $\m_i$ is $\erf\left(\frac{\ve}{\sqrt{2}}\right)$.
  
  If we take a single sample from the mixture, it will satisfy the desired conditions with probability at least
  $w_i\erf\left(\frac{\ve}{\sqrt{2}}\right)$.
  If we take $\frac{20\sqrt{2}}{3w_i\ve}$ samples from the mixture, the probability that some sample satisfies the conditions is at least
  $$1 - \left(1 - w_i\erf\left(\frac{\ve}{\sqrt{2}}\right)\right)^{\frac{20\sqrt{2}}{3w_i\ve}}
  \geq 1 - \left(1 - w_i\cdot \frac{3}{4}\frac{\ve}{\sqrt2}\right)^{\frac{20\sqrt{2}}{3w_i\ve}}
  \geq 1 - e^{-5} \geq \frac{99}{100}$$
  where the first inequality is by noting that $\erf(x) \geq \frac{3}{4}x$ for $x \in [0,1]$.
\end{prevproof}

\begin{prevproof}{Proposition}{prop:west}
  $w_i \geq \ve$ implies $w_i \geq \frac{w_i + \ve}{2}$, and thus $\frac{2}{\hat w_i} \geq \frac{2}{w_i + \ve} \geq \frac{1}{w_i}.$
\end{prevproof}

\begin{prevproof}{Lemma}{lem:means}
\label{app:means}
  Aside from the size of the collection, the rest of the conclusions follow from Propositions \ref{prop:nearmean} and \ref{prop:west}.
  
  For a given $\hat w$, the number of candidates $\hat \m_1$ we consider is 
  $\frac{40\sqrt{2}}{3\hat w\ve}$.
  We sum this over all candidates for $\hat w$, namely, $\ve, 2\ve, \dots, 1 - \ve$, giving us
  $$\sum_{t=1}^{\frac{1}{\ve} - 1}\frac{40\sqrt{2}}{3k\ve^2} = \frac{40\sqrt{2}}{3\ve^2}H_{\frac{1}{\ve}-1}
  = O\left(\frac{\log{\ve^{-1}}}{\ve^2}\right)$$
  where $H_n$ is the $n$th harmonic number.
\end{prevproof}

\begin{prevproof}{Lemma}{lem:dev1}
  Let $Y$ be the nearest sample to $\hat \m_1$.
  From Lemma \ref{lem:closeptgmmrob}, with probability $\geq \frac{9}{10}$, $|Y - \hat \m_1| \in [\frac{c_3}{4}\s_1, (\sqrt{2} + \frac{c_3}{4})\s_1]$.

  We can generate candidates by rearranging the bounds to obtain
  $$\frac{Y}{\sqrt{2} + \frac{c_3}{4}}
  \leq \s_1
  \leq \frac{Y}{\frac{c_3}{4}}$$

  Applying Fact $\ref{fact:mult}$ and noting that $\frac{R}{L} = O(1)$, we conclude that we can grid over this range with $O(\frac{1}{\ve})$ candidates.\end{prevproof}

\begin{prevproof}{Lemma}{lem:lastparams}
  \label{app:lastparams}
  The proof follows the sketch outlined in Section \ref{sec:lastcomp}.
  We first use Proposition \ref{prop:kolcons} to construct an approximation $\hat F$ of the GMM $F$.
  Using Proposition \ref{prop:dksubtract}, we see that
  $\dk\left(\mathcal{N}(\m_2,\s_2^2), \frac{\hat F - \hat w^* \mathcal{N}(\hat \m_1^*, \hat \s_1^{*2})}{1-\hat w^*}\right) \leq \frac{O(\ve)}{1 - w}$.
  By Lemma \ref{lem:mono}, we can compute a distribution $\hat H$ such that $\dk(\mathcal{N}(\m_2,\s_2^2), \hat H) \leq \frac{O(\ve)}{1 - w}$.
  Finally, using the median and interquartile range and the guaranteed provided by Lemmas \ref{lem:med} and \ref{lem:iqr}, we can compute candidates $\hat \m_2^* \in \m_2 \pm O\left(\frac{\ve}{1-w}\right)\s_2$ and $\hat \s_2^* \in \left(1\pm O\left(\frac{\ve}{1-w}\right)\right)\s_2$ from $\hat H$, as desired.
\end{prevproof}

\begin{prevproof}{Proposition}{prop:dksubtract}
\begin{align*}
\dk\left(\mathcal{N}(\m_2,\s_2^2), \frac{\hat F - \hat w^* \mathcal{N}(\hat \m_1^*, \hat \s_1^{*2})}{1-\hat w^*}\right)
&= \frac{1}{1 - \hat w}\dk(\hat w^* \mathcal{N}(\hat \m_1^*, \hat \s_1^{*2}) + (1-\hat w^*)\mathcal{N}(\m_2,\s_2^2), \hat F) \\
&\leq \frac{1}{1 - \hat w}(\dk(\hat w^* \mathcal{N}(\hat \m_1^*, \hat \s_1^{*2}) + (1-\hat w^*)\mathcal{N}(\m_2,\s_2^2), F) + \dk(F, \hat F)) \\
&\leq \frac{1}{1 - \hat w}(\dtv(\hat w^* \mathcal{N}(\hat \m_1^*, \hat \s_1^{*2}) + (1-\hat w^*)\mathcal{N}(\m_2,\s_2^2), F) + O(\ve)) \\
&\leq \frac{1}{1 - \hat w}(|w - \hat w| + \dtv(\mathcal{N}(\m_1,\s_1),\mathcal{N}(\hat \m_1^*,\hat \s_1^{*2}) + O(\ve)) \\
&\leq \frac{O(\ve)}{1 - \hat w} \\
&\leq \frac{O(\ve)}{1 - w}
\end{align*}
The equality is a rearrangement of terms,
the first inequality is the triangle inequality,
the second inequality uses Fact \ref{fact:dkdtv},
and the third and fourth inequalities use Propositions \ref{prop:ddodtv} and \ref{prop:mixtv} respectively.\end{prevproof}

\subsection{Proof of Theorem~\ref{thm:list of candidates}}

\begin{prevproof}{Theorem}{thm:list of candidates}
  We produce two lists of candidates corresponding to whether $\min{(w, 1-w)} = \Omega(\ve)$ or not: 
  \begin{itemize}
  
  \item In the first case, combining Lemmas \ref{lem:means}, \ref{lem:dev1}, and \ref{lem:lastparams} and taking the Cartesian product of the resulting candidates for the mixture's parameters, we see that we can obtain a collection of $O\left(\frac{\log \ve^{-1}}{\ve^3}\right)$ candidate mixtures.
  With probability $\geq \frac{4}{5}$, this will contain a candidate $(\hat w^*, \hat \m_1^*, \hat \m_2^*, \hat \s_1^*, \hat \s_2^*)$ such that
  $\hat w \in w \pm O(\ve), \hat \m_i \in \m_i \pm O(\ve)\s_i$ for $i=1,2$, and $\hat \s_i \in (1\pm O(\ve))\s_i$ for $i=1,2$.
  Note that we can choose the hidden constants to be as small as necessary for Lemma $\ref{lem:paramest}$, and thus we can obtain the desired total variation distance.
  
  Finally, note that the number of samples that we need for the above to hold is $O(1/\ve^2)$. For this, it is crucial that we {\em first} draw a sufficient $O(1/\ve^2)$ samples from the mixture (specified by the worse requirement among Lemmas~\ref{lem:means}, \ref{lem:dev1}, and \ref{lem:lastparams}), and {\em then} execute the candidate generation algorithm outlined in Lemmas~\ref{lem:means}, \ref{lem:dev1}, and \ref{lem:lastparams}. In particular, we do not want to redraw samples for every branching of this algorithm.
  
 Finally, to boost the success probability, we repeat the entire process $\log_{5} \d^{-1}$ times and let our collection of candidate mixutres be the union of the collections from these repetitions.
  The probability that none of these collections contains a suitable candidate distribution is $\leq \left(\frac15\right)^{\log_{5} \d^{-1}} \leq \d$.

\item In the second case, i.e. when one of the weights, w.l.o.g. $w_2$, is $O(\ve)$, we set $\hat{w}_1=1$ and we only produce candidates for $(\mu_1,\sigma_1^2)$. 
  Note that this scenario fits into the framework of Lemmas \ref{lem:med} and \ref{lem:iqr}.
  Our mixture $F$ is such that $\dk(F,\mathcal{N}(\m_1,\s_1^2)) \leq \dtv(F,\mathcal{N}(\m_1,\s_1^2)) \leq O(\ve)$.
  By the DKW inequality (Theorem \ref{thm:dkw}), we can use $O(\frac{1}{\ve^2} \cdot \log{\frac{1}{\d}})$ samples to generate the empirical distribution, which gives us a distribution $\hat F$ such that $\dtv(\hat F,\mathcal{N}(\m_1,\s_1^2)) \leq O(\ve)$ (by triangle inequality), with probability $\geq 1 - \d$.
  From this distribution, using the median and interquartile range and the guarantees of Lemmas \ref{lem:med} and \ref{lem:iqr}, we can extract $\hat \m_1^*$ and $\hat \s_1^*$ such that
  $|\hat \m_1^* - \m_1| \leq O(\ve)\s_1$ and 
  $|\hat \s_1^* - \s_1| \leq O(\ve)\s_1$.
  Thus, by Lemma \ref{lem:paramest}, we can achieve the desired total variation distance.
  
\end{itemize}

\end{prevproof}

%% file: intervalpartition.tex
\section{Details about Representing and Manipulating CDFs}
\label{app:intervalpartition}
First, we remark that, given a discrete random variable $X$ over a support of size $n$, in $\tilde O(n)$ time, we can construct a data structure which allows us to compute $F_X^{-1}(x)$ for $x \in [0,1]$ at the cost of $O(\log n)$ time per operation.
This data structures will be a set of disjoint intervals which form a partition of $[0,1]$, each associated with a value.
A value $x \in [0,1]$ will be mapped to the value associated to the interval which contains $x$.
This data structure can be constructed by mapping each value to an interval of width equal to the probability of that value.
This data structure can be queried in $O(\log n)$ time by performing binary search on the left endpoints of the intervals.
We name this the \emph{$n$-interval partition} representation of a distribution.
To avoid confusion with intervals that represent elements of the $\s$-algebra of the distribution, we refer to the intervals that are stored in the data structure as \emph{probability intervals}.

We note that, if we are only concerned with sampling, the order of the elements of the support is irrelevant.
However, we will sort the elements of the support in order to perform efficient modifications later.

At one point in our learning algorithm, we will have a candidate which correctly describes one of the two components in our mixture of Gaussians.
If we could ``subtract out'' this component from the mixture, we would be left with a single Gaussian - in this setting, we can efficiently perform parameter estimation to learn the other component.
However, if we naively subtract the probability densities, we will obtain negative probability densities, or equivalently, non-monotonically increasing CDFs.
To deal with this issue, we define a process we call monotonization.
Intuitively, this will shift negative probability density to locations with positive probability density.
We show that this preserves Kolmogorov distance and that it can be implemented efficiently.

\begin{definition}
  Given a bounded function $f : \mathbb{R} \rightarrow \mathbb{R}$,
  the \emph{monotonization} of $f$ is $\hat f$, where $\hat f(x) = \sup_{y \leq x} f(x)$.
\end{definition}

We argue that if a function is close in Kolmogorov distance to a monotone function, then so is its monotonization.
\begin{proposition}
  \label{prop:monoapprox}
  Suppose we have two bounded functions $F$ and $G$ such that $\dk(F,G) \leq \ve$, where $F$ is monotone non-decreasing.
  Then $\hat G$, the monotonization of $G$, is such that $\dk(F,\hat G) \leq \ve$.
\end{proposition}
\begin{proof}
  We show that $|F(x) - \hat G(x)| \leq \ve$ holds for an arbitrary point $x$, implying that $\dk(F,\hat G) \leq \ve$.  
  There are two cases: $F(x) \geq \hat G(x)$ and $F(x) < \hat G(x)$.

  If $F(x) \geq \hat G(x)$, using the fact that $\hat G(x) \geq G(x)$ (due to monotonization), we can deduce
  $|F(x) - \hat G(x)| \leq |F(x) - G(x)| \leq \ve$.

  If $F(x) < \hat G(x)$, consider an infinite sequence of points $\{y_i\}$ such that $G(y_i)$ becomes arbitrarily close to $\sup_{y \leq x} G(x)$.
  By monotonicity of $F$, we have that $|F(x) - \hat G(x)| \leq |F(y_i) - G(y_i)| + \d_i \leq \ve + \d_i$, where $\d_i = |\hat G(x) - G(y_i)|$.
  Since $\d_i$ can be taken arbitrarily small, we have $|F(x) - \hat G(x)| \leq \ve$.
\end{proof}

We will need to efficiently compute the monotonization in certain settings, when subtracting one monotone function from another.
\begin{proposition}
  \label{prop:comppart}
  Suppose we have access to the $n$-interval partition representation of a CDF $F$.
  Given a monotone non-decreasing function $G$, we can compute the $n$-interval partition representation of the monotonization of $F - G$ in $O(n)$ time.
\end{proposition}
\begin{proof}
  Consider the values in the $n$-interval partition of $F$.
  Between any two consecutive values $v_1$ and $v_2$, $F$ will be flat, and since $G$ is monotone non-decreasing, $F-G$ will be monotone non-increasing.
  Therefore, the monotonization of $F-G$ at $x \in [v_1,v_2)$ will be the maximum of $F-G$ on $(-\infty,v_1]$.
  The resulting monotonization will also be flat on the same intervals as $F$, so we will only need to update the probability intervals to reflect this monotonization.

  We will iterate over probability intervals in increasing order of their values, and describe how to update each interval.
  We will need to keep track of the maximum value of $F-G$ seen so far.
  Let $m$ be the maximum of $F-G$ for all $x \leq v$, where $v$ is the value associated with the last probability interval we have processed.
  Initially, we have the value $m = 0$.
  Suppose we are inspecting a probability interval with endpoints $[l,r]$ and value $v$.
  The left endpoint of this probability interval will become $\hat l = m$, and the right endpoint will become $\hat r = r - G(v)$.
  If $\hat r \leq \hat l$, the interval is degenerate, meaning that the monotonization will flatten out the discontinuity at $v$ - therefore, we simply delete the interval.
  Otherwise, we have a proper probability interval, and we update $m = \hat r$.

  This update takes constant time per interval, so the overall time required is $O(n)$.
\end{proof}

To conclude, we prove Lemma \ref{lem:mono}.

\begin{prevproof}{Lemma}{lem:mono}
First, by assumption, we know that $\frac{1}{1-w}\dk((1-w)H, F - wG) \leq \ve$.
By Proposition \ref{prop:comppart}, we can efficiently compute the monotonization of $F - wG$ - name this $(1-w)\hat H$.
By Proposition \ref{prop:monoapprox}, we have that $\frac{1}{1-w}\dk((1-w)H, (1-w)\hat H) \leq \ve$.
Renormalizing the distributions gives the desired approximation guarantee.

To justify the running time of this procedure, we must also argue that the normalization can be done efficiently.
To normalize the distribution $(1-w)\hat H$, we make another $O(n)$ pass over the probability intervals and multiply all the endpoints by $\frac{1}{r^*}$, where $r^*$ is the right endpoint of the rightmost probability interval.
We note that $r^*$ will be exactly $1-w$ because the value of $F - wG$ at $\infty$ is $1-w$, so this process results in the distribution $\hat H$.
\end{prevproof}

%% file: closestpointkol.tex
\section{Robust Estimation of Scale from a Mixture of Gaussians}
\label{sec:closestpointkol}

In this section, we examine the following statistic:

Given some point $x \in \mathbb{R}$ and $n$ IID random variables $X_1, \dots, X_n$, what is the minimum distance between $x$ and any $X_i$?

We give an interval in which this statistic is likely to fall (Proposition \ref{thm:closept}), and examine its robustness when sampling from distributions which are statistically close to the distribution under consideration (Proposition \ref{thm:closeptrob}).
We then apply these results to mixtures of Gaussians (Proposition \ref{thm:closeptgmm} and Lemma \ref{lem:closeptgmmrob}).

\begin{proposition}
  \label{thm:closept}
  Suppose we have $n$ IID random variables $X_1, \dots, X_n \sim X$, some $x \in \mathbb{R}$, and $y = F_X(x)$.
  Let $I_N$ be the interval $[F_X^{-1}(y - \frac{c_1}{n}), F_X^{-1}(y + \frac{c_1}{n})]$ and
      $I_F$ be the interval $[F_X^{-1}(y - \frac{c_2}{n}), F_X^{-1}(y + \frac{c_2}{n})]$ for some constants $0 < c_1 < c_2 \leq n$, and $I = I_F \backslash I_N$.
      Let $j = \arg\min_i |X_i - x|$.
      Then $\Pr[X_j \in I] \geq \frac{9}{10}$ for all $n > 0$.
\end{proposition}
\begin{proof}
  We show that $\Pr[X_j \not \in I] \leq \frac{1}{10}$ by showing that the following two bad events are unlikely:
  \begin{enumerate}
    \item We have a sample which is too close to $x$
    \item All our samples are too far from $x$
  \end{enumerate}
  Showing these events occur with low probability and combining with the union bound gives the desired result.

  Let $Y$ be the number of samples at distance $< \frac{c_1}{n}$ in distance in the CDF, i.e., $Y = |\{i\mid|F_X^{-1}(X_i) - y| < \frac{c_1}{n}\}|$.
  By linearity of expectation, $E[Y] = 2c_1$.
  By Markov's inequality, $\Pr(Y > 0) < 2c_1$.
  This allows us to upper bound the probability that one of our samples is too close to $x$.

  Let $Z$ be the number of samples at distance $< \frac{c_2}{n}$ in distance in the CDF, i.e., $Z = |\{i\mid|F_X^{-1}(X_i) - y| < \frac{c_2}{n}\}|$, and let $Z_i$ be an indicator random variable which indicates this property for $X_i$.
  We use the second moment principle,
  $$\Pr(Z > 0) \geq \frac{E[Z]^2}{E[Z^2]}$$
  By linearity of expectation, $E[Z]^2 = 4c_2^2$.
  \begin{align*}
    E[Z^2] &= \sum_i E[Z_i^2] + \sum_i \sum_{j \neq i} E[Z_iZ_j] \\
           &= 2c_2 + n(n-1)\left(\frac{4c_2^2}{n^2}\right) \\
           &\geq 2c_2 + 4c_2^2
  \end{align*}
  And thus, $\Pr(Z = 0) \leq \frac{1}{2c_2+1}$.
  This allows us to upper bound the probability that all of our samples are too far from $x$.

  Setting $c_1 = \frac{1}{40}$ and $c_2 = \frac{19}{2}$ gives probability $< \frac{1}{20}$ for each of the bad events, resulting in a probability $< \frac{1}{10}$ of either bad event by the union bound, and thus the desired result.
\end{proof}

We will need the following property of Kolmogorov distance, which states that probability mass within every interval is approximately preserved:
\begin{proposition}
  \label{prop:approxkol}
  If  $\dk(f_X,f_Y) \leq \ve$, then
  for all intervals $I \subseteq \mathbb{R}$, $|f_X(I) - f_Y(I)| \leq 2\ve$.
\end{proposition}
\begin{proof}
  For an interval $I = [a,b]$, we can rewrite the property as
  \begin{align*}
    |f_X(I) - f_Y(I)| &= |(F_X(b) - F_X(a)) - (F_Y(b) - F_Y(a))|  \\
                     &\leq |F_X(b) - F_Y(b)| + |F_X(a) - F_Y(a)| \\
                     &\leq 2\ve
  \end{align*}
  as desired, where the first inequality is the triangle inequality and the second inequality is due to the bound on Kolmogorov distance.
\end{proof}

The next proposition says that if we instead draw samples from a distribution which is close in total variation distance, the same property approximately holds with respect to the original distribution.

\begin{proposition}
  \label{thm:closeptrob}
  Suppose we have $n$ IID random variables $\hat X_1, \dots, \hat X_n \sim \hat X$ where $\dk(f_X,f_{\hat X}) \leq \d$, some $x \in \mathbb{R}$, and $y = F_X(x)$.
  Let $I_N$ be the interval $[F_X^{-1}(y - \frac{c_1}{n} + \d), F_X^{-1}(y + \frac{c_1}{n} - \d)]$ and
      $I_F$ be the interval $[F_X^{-1}(y - \frac{c_2}{n} - \d), F_X^{-1}(y + \frac{c_2}{n} + \d)]$ for some constants $0 < c_1 < c_2 \leq n$, and $I = I_F \backslash I_N$.
      Let $j = \arg\min_i |X_i - a|$.
      Then $\Pr[X_j \in I] \geq \frac{9}{10}$ for all $n > 0 $.
\end{proposition}
\begin{proof}
  First, examine interval $I_N$.
  This interval contains $\frac{2c_1}{n} - 2\d$ probability measure of the distribution $F_X$.
  By Proposition \ref{prop:approxkol}, $|F_X(I_N) - F_{\hat X}(I_N)| \leq 2\d$, so $F_{\hat X}(I_N) \leq \frac{2c_1}{n}$.
  One can repeat this argument to show that the amount of measure contained by $F_{\hat X}$ in $[F_X^{-1}(y - \frac{c_2}{n} - \d),F_X^{-1}(y + \frac{c_2}{n} + \d)]$ is $\geq \frac{2c_2}{n}$.

  As established through the proof of Proposition \ref{thm:closept}, with probability $\geq \frac{9}{10}$, 
  there will be no samples in a window containing probability measure $\frac{2c_1}{n}$, 
  but there will be at least one sample in a window containing probability measure $\frac{2c_2}{n}$. 
  Applying the same argument to these intervals, we can arrive at the desired result.
\end{proof}

We examine this statistic for some mixture of $k$ Gaussians with PDF $f$ around the point corresponding to the mean of the component with the minimum value of $\frac{\s_i}{w_i}$.
Initially, we assume that we know this location exactly and that we are taking samples according to $f$ exactly.

\begin{proposition}
  \label{thm:closeptgmm}
  Consider a mixture of $k$ Gaussians with PDF $f$, components $\mathcal{N}(\m_1,\s_1^2), \dots, \mathcal{N}(\m_k, \s_k^2)$ and weights $w_1, \dots, w_k$.
  Let $j = \arg\min_i \frac{\s_i}{w_i}$.
  If we take $n$ samples $X_1,\dots,X_n$ from the mixture (where $n > \frac{3\sqrt{\p}c_2}{2w_j}$), then
  $\min_i |X_i - \m_j| \in \left[\frac{\sqrt{2\p}c_1\s_j}{kw_jn},\frac{3\sqrt{2\p}c_2\s_j}{2w_jn}\right]$ with probability $\geq \frac{9}{10}$, where $c_1$ and $c_2$ are as defined in Proposition \ref{thm:closept}.
\end{proposition}
\begin{proof}
  We examine the CDF of the mixture around $\m_i$.
  Using Proposition 1 (and symmetry of a Gaussian about its mean), it is sufficient to show that
  $$\left[\m_i + \frac{\sqrt{2\p}c_1\s_i}{kw_in},\m_i + \frac{3\sqrt{2\p}c_2\s_i}{2w_in}\right] \supseteq
\left[F^{-1}(F(\m_i) + \frac{c_1}{n}), F^{-1}(F(\m_i) + \frac{c_2}{n})\right],$$ where $F$ is the CDF of the mixture.
  We show that each endpoint of the latter interval bounds the corresponding endpoint of the former interval.

  First, we show $\frac{c_1}{n} \geq F\left(\m_i +\frac{\sqrt{2\p}c_1\s_i}{kw_in}\right) -  F(\m_i)$.
  Let $I = \left[\m_i, \m_i + \frac{\sqrt{2\p}c_1\s_i}{kw_in}\right]$, $f$ be the PDF of the mixture, and $f_i$ be the PDF of component $i$ of the mixture.
  The right-hand side of the inequality we wish to prove is equal to 
  \begin{align*}
    \int_{I}\! f(x)\, \mathrm{d}x &= \int_{I}\! \sum_{j=1}^k w_jf_j(x)\, \mathrm{d}x \\
                                  &\leq \int_{I}\! \sum_{j=1}^k w_j\frac{1}{\s_j\sqrt{2\p}}\, \mathrm{d}x \\
                                  &\leq \int_{I}\! \frac{kw_i}{\s_i\sqrt{2\p}}\, \mathrm{d}x \\
                                  &= \frac{c_1}{n}
  \end{align*}
  where the first inequality is since the maximum of the PDF of a Gaussian is $\frac{1}{\s\sqrt{2\p}}$, and the second is since $\frac{\s_j}{w_j} \leq \frac{\s_i}{w_i}$ for all $j$.

  Next, we show $\frac{c_2}{n} \leq F\left(\m_i + \frac{3\sqrt{2\p}c_2\s_i}{2w_in}\right) -  F(\m_i)$.
  We note that the right-hand side is the probability mass contained in the interval - a lower bound for this quantity is the probability mass contributed by the particular Guassian we are examining, which is $\frac{w_i}{2}\erf{\left(\frac{3\sqrt{\p}c_2}{2w_in}\right)}$.
  Taking the Taylor expansion of the error function gives
  $$\erf(x) = \frac{2}{\sqrt{\p}}\left(x -\frac{x^3}{3} + O(x^5)\right) \geq \frac{2}{\sqrt{\p}}\left(\frac{2}{3}x\right)$$
  if $x < 1$.
  Applying this here, we can lower bound the contributed probability mass by $\frac{w_i}{2}\frac{2}{\sqrt{\p}}\frac{2}{3}\frac{3\sqrt{\p}c_2}{2w_in} = \frac{c_2}{n}$, as desired.
\end{proof}

Finally, we deal with uncertainties in parameters and apply the robustness properties of Proposition \ref{thm:closeptrob} in the following lemma:

\begin{prevproof}{Lemma}{lem:closeptgmmrob}
  We analyze the effect of each uncertainty:
  \begin{itemize}
    \item First, we consider the effect of sampling from $\hat f$, which is $\d$-close to $f$.
          By using Proposition \ref{thm:closeptrob}, we know that the nearest sample to $\m_j$ 
          will be at CDF distance between $\frac{c_1}{n} - \d \geq \frac{c_1}{2n}$ 
          and $\frac{c_2}{n} + \d \leq \frac{3c_2}{2n}$.
          We can then repeat the proof of Proposition \ref{thm:closeptgmm} 
          with $c_1$ replaced by $\frac{c_1}{2}$ and $c_2$ replaced by $\frac{3c_2}{2}$.
          This gives us that 
          $\min_i |X_i - \m_j| \in 
          \left[\frac{\sqrt{\p}c_1}{\sqrt{2}kw_jn}\s_j, \frac{9\sqrt{\p}c_2}{2\sqrt{2}w_jn}\s_j\right]$
          (where $n \geq \frac{9\sqrt{\p}c_2}{4w_j}$)
          with probability $\geq \frac{9}{10}$.
    \item Next, substituting in the bounds $\frac12\hat w_j \leq w_j \leq 2\hat w_j$, we get
          $\min_i \left|X_i - \m_j\right| \in 
          \left[\frac{\sqrt{\p}c_1}{2\sqrt{2}k\hat w_jn}\s_j, \frac{9\sqrt{\p}c_2}{\sqrt{2}\hat w_jn}\s_j\right]$
          (where $n \geq \frac{9\sqrt{\p}c_2}{2\hat w_j}$)
          with probability $\geq \frac{9}{10}$.
    \item We use $n = \frac{9\sqrt{\p}c_2}{2\hat w_j}$ samples to obtain:
          $\min_i \left|X_i - \m_j\right| \in 
          \left[\frac{c_3}{k}\s_j, \sqrt{2}\s_j\right]$
          with probability $\geq \frac{9}{10}$.
    \item Finally, applying $| \hat \m_j - \m_j| \leq \frac{c_3}{2k}\s_j$ gives the lemma statement.
  \end{itemize}
\end{prevproof}

%% file: tournament_appendix.tex
\section{Omitted Proofs from Section~\ref{sec:tournament}} \label{sec: tournament appendix}

\begin{prevproof}{Lemma}{lem:choosehypothesis}
We set up a competition between $H_1$ and $H_2$, in terms of the following subset of $\cal D$:
$$
{\cal W}_1\equiv{\cal W}_1(H_1,H_2) := \left\{w \in {\cal D}~\vline~H_1(w) > H_2(w) \right\}. \label{eq:W1}
$$

\noindent In terms of ${\cal W}_1$ we define $p_1 = H_1({\cal W}_1)$ and $p_2 = H_2({\cal W}_1)$. Clearly, $p_1 > p_2$ and
 $\dtv(H_1, H_2) = p_1-p_2$. The competition between $H_1$ and $H_2$ is carried out as follows:

\bigskip

\noindent
\framebox{
\medskip \noindent \begin{minipage}{16.2cm}
\medskip

\begin{enumerate}
\item[1a.] Draw $m=O\left({\log(1/\delta) \over \ve^2}\right)$ samples $s_1,\ldots,s_m$ from $X$, and let $\hat{\tau} = {1 \over m} | \{i~|~s_i
\in {\cal W}_1 \}|$ be the fraction of  them that fall inside ${\cal W}_1.$ \label{algorithm:step 1a}

\item[1b.] Similarly, draw $m$ samples from $H_1$, and let $\hat{p}_1$ be the fraction of  them that fall inside ${\cal W}_1.$ \label{algorithm:step 1b}

\item[1c.] Finally, draw $m$ samples from $H_2$, and let $\hat{p}_2$ be the fraction of  them that fall inside ${\cal W}_1.$ \label{algorithm:step 1c}

\item[2.] If $\hat{p}_1-\hat{p}_2\leq 6 \ve$, declare a draw. Otherwise:

\item[3.] If $\hat{\tau} > \hat{p}_1- 2 \ve$, declare $H_1$ as winner and return $H_1$; otherwise,

\item[4.] if $\hat{\tau} < \hat{p}_2+ 2 \ve$, declare $H_2$ as winner and return $H_2$; otherwise,

\item[5.] Declare a draw.
\end{enumerate}
\end{minipage}}

\medskip Notice that, in Steps~1a, 1b and 1c, the algorithm utilizes the PDF comparator for distributions $H_1$ and $H_2$. The correctness of the algorithm is a consequence of the following claim.

\begin{claim} \label{lem:kostas3}
Suppose that $\dtv(X,H_1) \leq \ve$. Then:
\begin{enumerate}
\item  If $\dtv(X, H_2)>8 \ve$, then the probability that the competition between $H_1$ and $H_2$ does not declare $H_1$ as the winner is at most $6e^{- {m \ve^2 /2 } }$;
\item If $\dtv(X, H_2)>4 \ve$, then the probability that the competition between $H_1$ and $H_2$ returns $H_2$ as the winner is at most $6e^{- {m \ve^2 /2 } }$.
\end{enumerate}
The analogous conclusions hold if we interchange $H_1$ and $H_2$ in the above claims. Finally, if $\dtv(H_1, H_2) \le 5 \ve$, the algorithm will declare a draw with probability at least $1-6e^{- {m \ve^2 /2 } }$.
\end{claim}
\begin{prevproof}{Claim}{lem:kostas3}
Let $\tau=X({\cal W}_1)$. The Chernoff bound (together with a union bound) imply that, with probability at least $1-6 e^{- {m \ve^2/2 }}$, the following are simultaneously true: $|p_1-\hat{p}_1| < \ve/2$, $|p_2-\hat{p}_2| < \ve/2$, and $|\tau-\hat{\tau}| < \ve/2$. Conditioning on these:
\begin{itemize}

\item If $\dtv(X,H_1) \leq \ve$ and $\dtv(X, H_2)>8 \ve$, then from the triangle inequality we get that  $p_1-p_2=\dtv(H_1, H_2) > 7 \ve$, hence $\hat{p}_1 - \hat{p}_2 > p_1-p_2 -\ve> 6 \ve$.  Hence, the algorithm will go beyond Step 2. Moreover, $\dtv(X,H_1) \leq \ve$ implies that $|\tau-p_1| \le \ve$, hence $|\hat{\tau}-\hat{p}_1| < 2 \ve$. So the algorithm will stop at Step 3, declaring $H_1$ as the winner of the competition between $H_1$ and $H_2$.

\item  If $\dtv(X,H_2) \leq \ve$ and $\dtv(X, H_1)>8 \ve$, then as in the previous case we get from the triangle inequality that  $p_1-p_2=\dtv(H_1, H_2) > 7 \ve$, hence $\hat{p}_1 - \hat{p}_2 > p_1-p_2 -\ve> 6 \ve$.  Hence, the algorithm will go beyond Step 2. Moreover, $\dtv(X,H_2) \leq \ve$ implies that $|\tau-p_2| \le \ve$, hence $|\hat{\tau}-\hat{p}_2| < 2 \ve$. So $\hat{p}_1 > \hat{\tau}+4 \ve$. Hence, the algorithm will not stop at Step 3, and it will stop at Step 4 declaring $H_2$ as the winner of the competition between $H_1$ and $H_2$.

\item If $\dtv(X,H_1) \leq \ve$ and $\dtv(X, H_2)>4 \ve$, we distinguish two subcases. If $\hat{p}_1 - \hat{p}_2 \le 6 \ve$, then the algorithm will stop at Step 2 declaring a draw. If $\hat{p}_1 - \hat{p}_2 > 6 \ve$, the algorithm proceeds to Step 3. Notice that $\dtv(X,H_1) \leq \ve$ implies that $|\tau-p_1| \le \ve$, hence $|\hat{\tau}-\hat{p}_1| < 2 \ve$. So the algorithm will stop at Step 3, declaring $H_1$ as the winner of the competition between $H_1$ and $H_2$.

\item If $\dtv(X,H_2) \leq \ve$ and $\dtv(X, H_1)>4 \ve$, we distinguish two subcases. If $\hat{p}_1 - \hat{p}_2 \le 6 \ve$, then the algorithm will stop at Step 2 declaring a draw. If $\hat{p}_1 - \hat{p}_2 > 6 \ve$, the algorithm proceeds to Step 3. Notice that $\dtv(X,H_2) \leq \ve$ implies that $|\tau-p_2| \le \ve$, hence $|\hat{\tau}-\hat{p}_2| < 2 \ve$. Hence, $\hat{p}_1 > \hat{p}_2  + 6 \ve \ge \hat{\tau} + 4 \ve$, so the algorithm will not stop at Step 3 and will proceed to Step 4. Given that $|\hat{\tau}-\hat{p}_2| < 2 \ve$, the algorithm will stop at Step 4, declaring $H_2$ as the winner of the competition between $H_1$ and $H_2$.

\item If $\dtv(H_1,H_2) \leq 5 \ve$, then $p_1-p_2 \le 5 \ve$, hence $\hat{p}_1-\hat{p}_2 \le 6 \ve$. So the algorithm will stop at Step 2 declaring a draw.
\end{itemize}
\end{prevproof}\end{prevproof}

\begin{prevproof}{Lemma}{thm:slow tournament theorem}
Draw $m=O(\log(2N/\delta) /\ve^2)$ samples from each of $X, H_1,\ldots,H_N$ and, using the same samples, run $$\text{{\tt ChooseHypothesis}}\left(X,H_i,H_j,\ve, {\delta \over 2N}\right),$$
for every pair of distributions $H_i, H_j \in {\cal H}$. If there is a distribution $H \in {\cal H}$ that was never a loser (but potentially tied with some distributions), output any such distribution. Otherwise, output ``failure.''

\medskip We analyze the correctness of our proposed algorithm in two steps. First, suppose there exists $H^* \in {\cal H}$ such that $\dtv(H^*,X) \leq \ve$. We argue that, with probability at least $1-{\delta \over 2}$, $H^*$ never
loses a competition against any other $H' \in \mathcal{H}$ (so the tournament does not output ``failure'').  Consider any $H'
\in {\cal H}$. If $\dtv(X, H') > 4 \ve$, by Lemma~\ref{lem:choosehypothesis} the probability that $H^*$ is not declared a winner or tie against 
$H'$ is at most ${\delta \over 2N}$. On the other hand, if $\dtv(X, H') \leq 4 \ve$, the triangle
inequality gives that $\dtv(H^*, H') \leq 5 \ve$ and, by Lemma~\ref{lem:choosehypothesis}, the probability that $H^*$ does not draw against $H'$ is at most $\delta \over 2N$.  A union
bound over all $N$ distributions in ${\cal H}$ shows that with probability at least $1-{\delta \over 2}$, the distribution $H^*$ never loses a competition.

We next argue that with probability at least $1-{\delta\over 2}$, every distribution $H \in \mathcal{H}$ that never loses must be $8\ve$-close to $X$. Fix a distribution $H$ such that $\dtv(X, H) > 8 \ve$. Lemma~\ref{lem:choosehypothesis} implies that
$H$ loses to $H^*$ with probability at least $1 - \delta/2N$.  A union bound gives that with
probability at least $1-{\delta \over 2}$, every distribution $H$ such that $\dtv(X, H) > 8 \ve$ loses some competition.

Thus, with overall probability at least $1-\delta$, the tournament does not output ``failure'' and outputs some distribution $H$ such that $\dtv(H, X) \le 8 \ve.$ \end{prevproof}

\begin{prevproof}{Claim}{claim1:tournament}
The probability that ${\cal H}'$ contains no distribution that is $8 \ve$-close to $X$ is at most
$$(1-p)^{\lceil {3 \sqrt{N}} \rceil } \le e^{-3}.$$
If ${\cal H}'$ contains at least one distribution that is $8 \ve$-close to $X$, then by Lemma~\ref{thm:slow tournament theorem} the distribution output by  {\tt SlowTournament}$(X,{\cal H}', 8\ve, e^{-3})$ is $64\ve$-close to $X$ with probability at least $1-e^{-3}$. From a union bound, it follows that the distribution output by {\rm S1} is $64 \ve$-close to $X$, with probability at least $1-2 e^{-3} \ge 9/10$. The bounds on the number of samples and operations follow from Lemma~\ref{thm:slow tournament theorem}.
\end{prevproof}

\begin{prevproof}{Claim}{claim2:tournament}
Suppose that there is some distribution $H^* \in {\cal H}$ that is $\ve$-close to $X$. We first argue that with probability at least ${1 \over 3}$, $H^* \in {\cal H}_{i_T}$. We show this in two steps:
\begin{itemize}
\item[(a)] Recall that we draw samples from $X, H_1,\ldots,H_N$ before Phase 1 begins, and reuse the same samples whenever required by some execution of {\tt ChooseHypothesis} during Phase 1. Fix a realization of these samples. We can ask the question of what would happen if we executed\\ {\tt ChooseHypothesis}$(X,H^*,{H}_j,\ve,1/3N)$, for some $H_j \in {\cal H} \setminus \{H^*\}$ using these samples. From Lemma~\ref{lem:choosehypothesis}, it follows that, if $H_j$ is farther than $8\ve$-away from $X$, then $H^*$ would be declared the winner by {\tt ChooseHypothesis}$(X,H^*,{H}_j,\ve,1/3N)$, with probability at least $1-1/3N$. By a union bound, our samples satisfy this property simultaneously for all $H_j \in {\cal H}\setminus \{H^*\}$ that are farther than $8\ve$-away from $X$, with probability at least $1-1/3$. Henceforth, we condition that our samples have this property.

\item[(b)] Conditioning on our samples having the property discussed in (a), we argue that $H^* \in {\cal H}_{i_T}$ with probability at least $1/2$ (so that, with overall probability at least $1/3$, it holds that $H^* \in {\cal H}_{i_T}$). It suffices to argue that, with probability at least $1/2$, in all iterations of Phase 1, $H^*$ is not matched with a distribution that is $8\ve$-close to $X$. This happens with probability at least:
$$(1-p)(1-2p) \cdots (1-2^{T-1}p) \ge  2^{-2p\sum_{i=0}^{T-1} 2^i} = 2^{-2p(2^T-1)}\ge 1/2.$$
Indeed, given the definition of $p$, the probability that $H^*$ is not matched to a distribution that is $8\ve$-close to $X$ is at least $1-p$ in the first iteration. If this happens, then (because of our conditioning from (a)), $H^*$ will survive this iteration. In the next iteration, the fraction of surviving distributions that are $8\ve$-close to $X$ and are different than $H^*$ itself is at most~$2p$. Hence, the probability that $H^*$ is not matched to a distribution that is $8\ve$-close to $X$ is at least $1-2p$ in the second iteration, etc.
\end{itemize}
Now, conditioning on $H^* \in {\cal H}_{i_T}$, it follows from Lemma~\ref{thm:slow tournament theorem} that the distribution $\hat{H}$ output by\\ {\tt SlowTournament}$(X,{\cal H}_{i_T}, \ve, 1/4)$ is $8\ve$-close to $X$ with probability at least $3/4$. 

Hence, with overall probability at least $1/4$, the distribution output by {\rm S2} is $8 \ve$-close to $X$.

\smallskip The number of samples drawn from each distribution in ${\cal H} \cup \{X\}$ is clearly $O({1\over \ve^2} \log N)$, as Phase 1 draws $O({1\over \ve^2} \log N)$ samples from each distribution and, by Lemma~\ref{thm:slow tournament theorem}, Phase 2 also draws $O({1\over \ve^2} \log N)$ samples from each distribution.

\smallskip The total number of operations is bounded by $O({1 \over \ve^2} N \log N).$
Indeed, Phase 1 runs {\tt ChooseHypothesis} $O(N)$ times, and by Lemma~\ref{lem:choosehypothesis} and our choice of $1/3N$ for the confidence parameter of each execution, each execution takes $O(\log N/\ve^2)$ operations. So the total number of operations of Phase 1 is $O({1 \over \ve^2} N \log N)$. On the other hand, the size of ${\cal H}_{i_T}$ is at most ${2^{\lceil \log_2 N \rceil} \over 2^{T}} = {2^{\lceil \log_2 N \rceil} \over 2^{\lfloor \log_2 {\sqrt{N} \over 2} \rfloor}} \le 8 \sqrt{N}$. So by Lemma~\ref{thm:slow tournament theorem}, Phase 2 takes $O({1 \over \ve^2} N \log N)$ operations.
\end{prevproof}

\section{Faster Slow and Fast Tournaments} \label{sec: tournament appendix 2}
In this section, we describe another hypothesis selection algorithm.
This algorithm is faster than {\tt SlowTournament}, though at the cost of a larger constant in the approximation factor.
In most reasonable parameter regimes, this algorithm is slower than {\tt FastTournament}, and still has a larger constant in the approximation factor.
Regardless, we go on to show how it can be used to improve upon the worst-case running time of {\tt FastTournament}.

\begin{theorem} \label{thm:recursive tournament theorem}
For any constant $\g > 0$, there is an algorithm {\tt RecursiveSlowTournament}$_{\g}(X, {\cal H},\ve,\delta)$, which is given sample access to some distribution $X$ and a collection of distributions ${\cal H}=\{H_1,\ldots,H_N\}$ over some set ${\cal D}$, access to a PDF comparator for every pair of distributions $H_i, H_j \in {\cal H}$, an accuracy parameter $\ve >0$, and a confidence parameter $\delta >0$.  The algorithm makes
$m=O(\log(N/\delta) /\ve^2)$ draws from each of $X, H_1,\ldots,H_N$ and returns some $H \in {\cal H}$ or declares ``failure.''  If there is some $H^* \in {\cal H}$ such that $\dtv(H^*,X) \leq \ve$ then with probability at least $1-\delta$ the distribution $H$ that {\tt
RecursiveSlowTournament} returns satisfies $\dtv(H,X) \leq O(\ve).$ The total number of operations of the algorithm is $O\left( N^{1+\g} \log(N/\delta) /\ve^2 \right)$.

\end{theorem}
\begin{proof}
For simplicity, assume that $\sqrt{N}$ is integer. (If not, introduce into ${\cal H}$ multiple copies of an arbritrary $H \in {\cal H}$ so that $\sqrt{N}$ becomes an integer.) Partition ${\cal H}$ into $\sqrt{N}$ subsets, ${\cal H} = {\cal H}_1 \sqcup {\cal H}_2 \sqcup \ldots \sqcup {\cal H}_{\sqrt{N}}$ and do the following:
\begin{enumerate}
\item Set $\delta'= \delta/2$, draw $O(\log(\sqrt{N}/\delta') /\ve^2)$ samples from $X$ and, using the same samples, run {\tt SlowTournament}$(X,{\cal H}_i, \ve,\delta')$ from Lemma~\ref{thm:slow tournament theorem} for each $i$;

\item Run {\tt SlowTournament}$(X,{\cal W}, 8\ve,\delta')$, where ${\cal W}$ are the distributions output by {\tt SlowTournament} in the previous step. If ${\cal W}=\emptyset$ output ``failure''.
\end{enumerate}
Let us call the above algorithm {\tt SlowTournament}$^{\bigotimes 1}(X,{\cal H}, \ve,\delta)$, before proceeding to analyze its correctness, sample and time complexity. Suppose there exists a distribution $H \in {\cal H}$ such that $\dtv(H,X) \le \ve$.  Without loss of generality, assume that $H \in {\cal H}_1$. Then, from Lemma~\ref{thm:slow tournament theorem}, with probability at least $1-\delta'$, {\tt SlowTournament}$(X,{\cal H}_1,\ve,\delta')$ will output a distribution $H'$ such that $\dtv(H',X) \le 8 \ve$. Conditioning on this and applying Lemma~\ref{thm:slow tournament theorem} again, with conditional probability at least $1-\delta'$ {\tt SlowTournament}$(X,{\cal W},8 \ve,\delta')$ will output a distribution $H''$ such that $\dtv(H'',X) \le 64 \ve$. So with overall probability at least $1-\delta$, {\tt SlowTournament}$^{\bigotimes 1}(X,{\cal H}, \ve,\delta)$ will output a distribution that is $64\ve$-close to $X$. The number of samples that the algorithm draws from $X$ is $O(\log(N/\delta) /\ve^2)$, and the running time is 
$$\sqrt{N} \times O\left(N \log(N/\delta') /\ve^2 \right) + O\left(N \log(N/\delta') /(8\ve)^2 \right) = O\left(N^{3/2} \log(N/\delta) /\ve^2 \right).$$
So, compared to {\tt SlowTournament}, {\tt SlowTournament}$^{\bigotimes 1}$ has the same sample complexity asymptotics and the same asymptotic guarantee for the distance from $X$ of the output distribution, but the exponent of $N$ in the running time improved from $2$ to $3/2$.

For $t=2, 3, \ldots$, define {\tt SlowTournament}$^{\bigotimes t}$ by replacing  {\tt SlowTournament} by {\tt SlowTournament}$^{\bigotimes t-1}$ in the code of {\tt SlowTournament}$^{\bigotimes 1}$. It follows from the same analysis as above that as $t$ increases the exponent of $N$ in the running time gets arbitrarily close to $1$. In particular, in one step an exponent of $1+\alpha$ becomes an exponent of $1+\alpha/2$. So for some constant $t$, {\tt SlowTournament}$^{\bigotimes t}$ will satisfy the requirements of the theorem.
\end{proof}

As a corollary, we can immediately improve the running time of {\tt FastTournament} at the cost of the constant in the approximation factor.
The construction and analysis is nearly identical to that of {\tt FastTournament}.
The sole difference is in step 3 of {\tt FastTournament}$_A$ - we replace {\tt SlowTournament} with {\tt RecursiveSlowTournament}$_\g$.

\begin{corollary} \label{corr:recursive tournament corollary}
For any constant $\g > 0$, there is an algorithm {\tt FastTournament}$_\g(X, {\cal H},\ve,\delta)$, which is given sample access to some distribution $X$ and a collection of distributions ${\cal H}=\{H_1,\ldots,H_N\}$ over some set ${\cal D}$, access to a PDF comparator for every pair of distributions $H_i, H_j \in {\cal H}$, an accuracy parameter $\ve >0$, and a confidence parameter $\delta >0$.  The algorithm makes
{$O\left({\log {1/ \delta} \over \ve^2} \cdot \log N\right)$} draws from each of $X, H_1,\ldots,H_N$ and returns some $H \in {\cal H}$ or declares ``failure.''  If there is some $H^* \in {\cal H}$ such that $\dtv(H^*,X) \leq \ve$ then with probability at least $1-\delta$ the distribution $H$ that {\tt
SlowTournament} returns satisfies $\dtv(H,X) \leq O(\ve)$. The total number of operations of the algorithm is {$O\left( {\log{1 / \delta} \over \ve^2} (N \log N + \log^{1+\g} {1 \over \delta}) \right)$}.
Furthermore, the expected number of operations of the algorithm is {$O\left( {N\log{N /\d} \over \ve^2}\right)$}.
\end{corollary}